\documentclass[10pt]{IEEEtran}
\usepackage{pgfplots}
\usepackage{gastex}

\newcommand{\zones}{\mathcal{Z}}

\usepackage{mathtools}
\DeclarePairedDelimiter{\fract}{\lbag}{\rbag}
\DeclarePairedDelimiter{\floor}{\lfloor}{\rfloor}

\usepackage{color}

\usepackage{tikz}
\usepackage{xcolor}
\usepackage{pgf}
\usepackage{pgflibraryarrows}
\usepackage{pgffor}
\usepackage{pgflibrarysnakes}
\usetikzlibrary{fit} % fitting shapes to coordinates
\usetikzlibrary{backgrounds} % drawing the background after the foreground
\usetikzlibrary{patterns}

\tikzstyle{background}=[rectangle,fill=gray!10, inner sep=0.1cm, rounded corners=0mm]
\usepgflibrary{shapes}
\usetikzlibrary{snakes,automata}
\usetikzlibrary{shadows}

\tikzstyle{nloc}=[draw, text badly centered, rectangle, rounded corners, minimum size=2em,inner sep=0.5em]
\tikzstyle{loc}=[draw,rectangle,minimum size=1.4em,inner sep=0em]
\tikzstyle{trans}=[-latex, rounded corners]
\tikzstyle{trans2}=[-latex, dashed, rounded corners]

\usepackage{amsmath}
\usepackage{amsfonts}
\usepackage{amssymb}
\usepackage{hyperref}

\newcommand{\until}{\:\mathcal{U}}
\usepackage{graphicx}
\usepackage{hyperref}
\usepackage{verbatim}

\makeatletter
\newif\if@restonecol
\makeatother

\usepackage[ruled,noend,lined]{algorithm2e}

\newtheorem{theorem}{Theorem}
\newtheorem{example}{Example}
\newtheorem{definition}{Definition}

\newtheorem{proposition}[theorem]{Proposition}

\newenvironment{proof}{\noindent{\em Proof.}}{\qed \par \smallskip }
\newcommand{\qed}{{\nopagebreak\vspace{1ex}
\nopagebreak\hspace*{\fill}\mbox{$\rule{5pt}{5pt}$}}}

\usepackage{stmaryrd}

\newcommand{\sem}[1]{\llbracket#1\rrbracket}

\newcommand{\pplus}[1]{{\mathop{\oplus}_{#1}}}

\newcommand{\post}{\text{\textsc{Post}}}

\newcommand{\tot}[1]{\mathit{T}(#1)}

\newcommand{\Kk}{\mathcal{K}}
\newcommand{\Tt}{\mathcal{T}}
\newcommand{\TS}{\mathcal{T}}
\newcommand{\Hh}{\mathcal{H}}
\newcommand{\Mm}{\mathcal{M}}

\newcommand{\Rr}{\mathcal{R}}
\newcommand{\Aa}{\mathcal{A}}

\newcommand{\Cc}{\mathcal{C}}

\newcommand{\jJ}{\mathbb{J}}
\newcommand{\mM}{\mathbb{M}}

\newcommand{\set}[1]{\left\{ #1 \right\}}

\newcommand{\seq}[1]{\langle #1 \rangle}
\newcommand{\Rplus}{{\mathbb R}_{\geq 0}}

\newcommand{\Real}{\mathbb R}

\newcommand{\Int}{\mathbb{Z}}

\def\rmdef{\stackrel{\mbox{\rm {\tiny def}}}{=}} % equals definition

\newcommand{\done}{\textrm{done}}
\newcommand{\beg}{\textrm{begin}}
\newcommand{\fin}{\textrm{finish}}

\newcommand{\pred}{\textrm{pred}}

\newcommand{\rect}{\textrm{rect}}

\newcommand{\rst}{\textrm{reset}}

\newcommand{\Trace}{\mathsf{\it Trace}}

\newcommand{\norm}[1]{\|#1\|}

\title{Hybrid Automata for Formal Modeling and Verification of Cyber-Physical
  Systems}

\author{
  \IEEEauthorblockN{Shankara Narayanan Krishna and Ashutosh Trivedi}

  \IEEEauthorblockA{Department of Computer Science and Engineering\\ Indian
    Institute of Technology Bombay\\Mumbai, India 400076\\
    Email: \{krishnas,trivedi\}@cse.iitb.ac.in}
}

\begin{document}
\date{\today}
\maketitle

\begin{abstract}
The presence of a tight integration between the discrete control (the ``cyber'')
and the analog environment (the ``physical'')---via sensors and actuators over
wired or wireless communication networks---is the defining feature of
cyber-physical systems. 
Hence, the functional correctness of a cyber-physical system is crucially
dependent not only on the dynamics of the analog physical environment, but also
on the decisions taken by the discrete control that alter the dynamics of the
environment.  
The framework of \emph{Hybrid automata}---introduced by Alur, Courcoubetis,
Henzinger, and Ho---provides a formal modeling and specification environment to
analyze the interaction between the discrete and continuous parts of a
cyber-physical system.
Hybrid automata can be considered as generalizations of finite state automata
augmented with a finite set of real-valued variables whose dynamics in each state is
governed by a system of ordinary differential equations.
Moreover, the discrete transitions of hybrid automata are guarded by constraints
over the values of these real-valued variables, and enable discontinuous jumps
in the evolution of these variables. 
Considering the richness of the dynamics in a hybrid automaton, it is perhaps
not surprising that the fundamental verification questions, like reachability
and schedulability, for the general model are undecidable. 
In this article we present a review of hybrid automata as modeling and
verification framework for cyber-physical systems, and survey some of the key
results related to practical verification questions related to hybrid automata.
\end{abstract}

\section{Introduction}
\label{sec:introduction}
The term ``cyber-physical systems'' refers to any network of digital and analog
systems whose performance crucially depends on both the continuous dynamics of
the analog parts and the real-time switching decisions made by the digital system. 
A typical cyber-physical system may consist of several processors connected with
a set of physical systems via sensors and actuators over wired or wireless
communication networks. 
Such systems are increasingly playing safety-critical role in modern life, where
a fault in their design can be catastrophic.  

Modern cars are an important paradigmatic example of such safety-critical
cyber-physical systems.  
A modern premium car typically has 70 to 100 interconnected electronic control
units (ECUs) with dozens of sensors~\cite{BKPS07} performing various
functions~\cite{C13} like air-bag control,  cruise control, electronic stability
control, antilock brakes, engine ignition, windshield-wiper control,  engine
control, and collision-avoidance system.  
Many of these ECUs are connected with analog environment via sensors and
actuators, and are expected to perform their operations within hard time limits.
For instance, the air-bag ECU needs to respond within 20-30 millisecond after
the impact sensor connected to it detects a severe impact.  
As the  number of ECUs in a typical car is increasing and performing more
autonomously,  it is becoming increasingly difficult to ensure their
correctness.   
The severity of the problem can perhaps be best realized by looking into the
growing list of recalls~\cite{C13} by leading car companies due to
software-related problems.  
Some prominent examples include Toyota's  recall of 160,000 of its 2004/05 Prius
models because of a  software problem causing the car to suddenly stall,
Jaguar's 2011 recall of nearly 18,000 X-type cars  due to a software bug
resulting in driver's inability in turning off the cruise control,  and
Volkswagen's 2011 recall of about 4000 of its 2008 Passats models for
engine-control module software problem.   
The list is long and underscores the challenges in designing and verifying
safety-critical cyber-physical systems. 
Similar examples can also be cited for the cyber-physical system from other
domains such as avionics, implantable medical devices, transportation networks,
and energy sector.  

Formal modeling and verification of systems is the set of techniques that employ
rigorous mathematical reasoning to analyze properties of a system. 
In this article we concentrate on a celebrated~\cite{PK0,PK1} formal
verification framework known as \emph{model checking}~\cite{CGP99}.
Model Checking---pioneered by Clarke, Sifakis and Emerson~\cite{TM1}---is a
widely used automated technique that, given a formal description of a system and
a property, systematically checks whether this property holds for a given state
of the system model.  
The three key steps of this framework are the following:
\begin{enumerate}
\item \emph{formal modeling}: modeling a system under consideration using
  mathematically precise syntax that approximate a given system to a desired
  level of abstraction; 
\item \emph{formal specification}: specify the properties of the system using a 
  mathematically precise specification language (typically in formal logic); and 
\item \emph{formal analysis}: analyze the formal model with respect to the formal
  specification and report counter-example in case the system model violates the
  specification. 
\end{enumerate}
The success of the model checking framework in formal verification of systems
is largely due to it being highly automatic---a push-button
technology~\cite{CES09}---in comparison to other competing approaches like theorem proving. 
The counterexamples generated in the model-checking process often are used to
automatically refine---known as counterexample-guided abstraction refinement
(CEGAR)~\cite{CGHLV00,CFHKOST03} framework---the model and/or the property and
the entire procedure can  be repeated and thus removing the need of a very
accurate initial model or specification. 

Early research on formal modeling and verification of systems concentrated on
simplified models of the systems as  finite state-transition graphs.
Since these models are finite in nature, it is---in theory---possible to
exhaustively explore the state space of the system to verify the properties of
interest.  
However, the biggest challenge in model-checking of finite state-transition
graphs is so-called \emph{state-space explosion problem}~\cite{CGP99}
characterizing the exponential blowup in the number of states in the explicit
representation of the system where the system is naturally represented
succinctly using state variables, or as a composition of a network of interacting
finite state-transition graphs. 
In general, the state-space explosion problem renders the explicit exhaustive
exploration of the system intractable. 
However, a number of techniques have been proposed to overcome the state-space
explosion problem---including symmetry reduction~\cite{CAJS98}, partial-order
reduction~\cite{Peled94}, symbolic model checking~\cite{McMilan93} and bounded model
checking~\cite{BCCSZ03,BCCSZ09}---that has culminated into efficient and mature
tool support including SPIN~\cite{TSPIN} and NuSMV~\cite{TNuSMV} for finite state
model-checking.  
Examples of the use of finite-state model-checking in industry include the
verification of hardware circuits~\cite{Kur08}, communication~\cite{ADMP01} and
security~\cite{MCJ97,BCM11} protocols, and software device drivers~\cite{BLR11}.  

These finite state-transition graphs, however, often do not satisfactorily model
cyber-physical systems as they disregard the continuous dynamics of the physical
environment.
Alur and Dill~\cite{AD90} were the first one to propose a formal model, known as
timed automata, combining finite state-transition graphs with a finite set of
real-valued variables that evolve as time progresses while the system occupies
a state.  
In a timed automaton the real-valued variables---called clocks---simulate
perfect clocks as they evolve with a uniform constant speed (rate) and hence can
model asynchronous real-time systems interacting with a continuous physical
environment.       
The clock variables can be used to constrain the evolution of the system by
guarding the transitions of the graph, and can also be reseted at the time of
taking a transition to remember the time since that transition.
These capabilities make timed automata quite expressive formalism to define
real-time systems.
Moreover, the decidability\footnote{
The concept of decidability is a central one in computer science and it
characterizes the set of problems for which one can write computer programs that
always terminate with a correct answer.   
The problems for which it is not possible to write such a program are known as
\emph{undecidable} problems.  
A most famous undecidable problem is the halting problem (similar to
reachability problem) for the configurations of Turing machines (an abstract
model of computation capturing the notion of algorithmic computation).
} of key verification problems like reachability and
schedulability~\cite{AD90} and availability of mature verification tools---like
UPPAAL~\cite{UPPAAL,TUppaal}, Kronos~\cite{TKronos}, and RED~\cite{TRED}---make
timed automata an appealing tool for real-time system verification.

Alur, Courcoubetis, Henzinger, and Ho generalized the timed automata to hybrid
automata~\cite{ACHH93} to include real-valued variables with arbitrary dynamics
specified using ordinary differential equations. 
Considering the richness of dynamics of a hybrid automata, it is perhaps not
surprising that the fundamental verification questions like reachability are
undecidable for hybrid automata.
A number of subclasses of hybrid automata has been proposed with decidable
verification problems and some of the algorithms have been implemented as part
of tools like HyTech~\cite{THyTech} and PHAVer~\cite{Tphaver}.

Timed and hybrid automata provide an intuitive and semantically unambiguous way
to model cyber-physical systems, and a number of
case-studies~\cite{UP01,CJLRR09,LTS08,FTY11,Pas06,SMF97,JPM12} 
demonstrate their application for the analysis of cyber-physical systems.  
In this article we aim to provide a general introduction to verification using
hybrid automata as we focus on model-checking classical LTL logic~\cite{MP92}
over hybrid automata. 
To keep the discussion simple we do not cover other logics, for instance,
computation tree logic (CTL, CTL$^*$)~\cite{MP92,CGP99}, modal
$\mu$-calculus~\cite{Eme96}, and real-time and hybrid extensions of these
logics~\cite{AH92} including metric temporal logics (MTL~\cite{koy90,OW08}) and
duration calculus (DC)~\cite{GNRR93}.  

The goal of this article is to introduce key concepts for cyber-physical systems
modeling and verification using hybrid automata with a focus on LTL model-checking. 
In order to better focus our attention, we will not cover several useful
extensions of hybrid automata that capture certain natural aspects of modeling
hybrid systems, including  
\begin{itemize}
\item  
  game-theoretic extensions~\cite{AMPS98,AM99,dAFHMS03,ABM04,CHP08,BBJLR08}
  that allow the model to distinguish between controllable and
  uncontrollable non-determinism; 
\item 
  probabilistic extensions~\cite{KNSS99,Bea03,KNPS06,AB06,BF09,MLK10} that
  permit modeling of stochastic behavior arising due to, e.g.,  faulty or
  unreliable sensors or actuators, uncertainty in timing delays, and performance
  characteristics of (third-party) components; and
\item  
  priced extensions~\cite{LBBFHPR01,BBL08,RLS06,JJK06,Bou06} that permit
  modeling of resource consumption and payoffs associated with decisions.
\end{itemize}
We also restrict our attention to theoretical results regarding decidability of
LTL model-checking problems, and  do not cover data structures and
algorithms~\cite{HHW97,Frehse05,UPPAAL} for efficient implementation of these
results. 

We begin (Section~\ref{sec:definitions}) this survey  by introducing two
formalisms to model discrete and continuous dynamical systems, and then we
present hybrid automata model that combines features from these two models.  
%We also present the concept of composition of network of hybrid automata that
%enables separation of concerns and modularity while modeling cyber-physical
%systems. 
Section~\ref{sec:verification} introduces syntax and semantics of linear
temporal logic (LTL) followed by a formal definition of corresponding
model-checking problem over a hybrid automata, and using two-counter Minsky
machines~\cite{Min67} we prove the in general LTL model-checking over hybrid
automata is undecidable.
In this section, we also introduce the idea of state-space reduction using a
well-established technique called \emph{quotienting} which we later exploit to
show decidability of model checking problem for some variants of hybrid
automata. 
We conclude the survey by discussing (Section~\ref{sec:decidability}) three key
subclasses of hybrid automata---timed automata, (initialized) rectangular hybrid
automata, and (two dimensional) piecewise-constant derivative systems---with
decidable model checking problem.   
%to prove decidability of LTL model-checking
%problem for some notable subclasses of hybrid automata including timed automata
%and initialized rectangular hybrid automata.
%We also present a brief review of other subclasses including piecewise-constant
%derivative systems~\cite{AMP95b} and constant-rate multi-mode
%systems~\cite{ATW12}. 
%We survey some of the key subclasses with decidable LTL model-checking problems, and
%review some well-known boundaries between decidable and undecidable classes of
%hybrid automata. 
%The rest of the article is organized as the following.
%In the next section 
%We introduce formal verification framework for hybrid automata in
%Section~\ref{sec:verification}, followed by Section~\ref{sec:decidability} where
%we discuss various subclasses of hybrid automata that allow algorithmic exploration
%of the state space to decide their verification problem. 
%We conclude the article with some remarks on how to improve applicability of formal
%verification framework to cyber-physical system development. 

\section{Hybrid Automata} 
\label{sec:definitions}
A \emph{dynamical system} is simply a system whose ``state'' evolves with
``time'' governed by a fixed set of rules or ``dynamics''.
The state of a dynamical system is specified as valuations of the variables of
interest in the system. 
Depending upon the nature of variables (discrete or continuous) and the notion
of time (discrete or continuous) the dynamics of variables can be specified by
differential equations or discrete assignments.  
For the purpose of this paper, we classify the dynamical systems into the
following three broad classes: i) \emph{discrete systems} where both
the notion of time and the variables are discrete, ii) \emph{continuous systems}
where the notion of time is continuous, while the variables are continuous, and
iii) \emph{hybrid systems} where some variables are continuous and some are
discrete, and although the notion of time is continuous, special
dynamic-changing events can happen at discrete instants.  
Notice that both discrete and continuous systems can be considered as subclasses
of hybrid systems.

On an abstract level any dynamical system can simply be represented as a graph
whose nodes represent the states and edges represent transition between the
states. 
Formally, a state transition graph can be defined in the following manner.
\begin{definition}[State Transition Graphs]
  A state transition graph is a tuple $\Tt = (S, S_0, \Sigma, \Delta)$
  where: 
  \begin{itemize}
  \item $S$ is a (potentially infinite) set of \emph{states}; 
  \item $S_0 \subseteq S$ is the set of \emph{initial states};
  \item $\Sigma$ is a (potentially infinite) set of \emph{actions}; and
  \item $\Delta \subseteq S \times \Sigma \times S$ is the \emph{transition
      relation}.
  \end{itemize}
  We say that a state transition graph $\Tt$ is finite (countable), if the sets
  $S$ and $\Sigma$ are finite (countable). 
\end{definition}
Given an action $a \in \Sigma$ and a state $s$ we write $\post(s,a)$ for the set
of states that are reachable from $s$ on $a$ and $\post(s)$ for the states
reachable in one step from $s$, i.e.
\begin{eqnarray*}
\post(s, a) &=& \set{s' \::\: (s,a,s') \in \Delta}\\
\post(s) &=& \bigcup_{a \in \Sigma} \post(s,a). 
\end{eqnarray*}
A run---an execution or a trajectory---of a dynamical system modeled as a state
transition graph $\Tt$ is a (finite or infinite) alternating sequence of states
and actions that begins with an initial state and all consecutive states are
connected with their predecessor via the transition relation. 
Formally, a finite run is a sequence $\seq{s_0, a_1, s_1, a_2, s_2, \ldots,
  s_n}$ such that  $s_0 \in S_0$ and for all $0 \leq i < n$ we have that
$s_{i+1}  \in \post(s_i, a_{i+1})$. 
An infinite run is defined analogously.

\begin{example}
  A graphical description of a state transition graph depicting a mod-4 counter
  with pause is shown in Figure~\ref{counter-sem}. 
  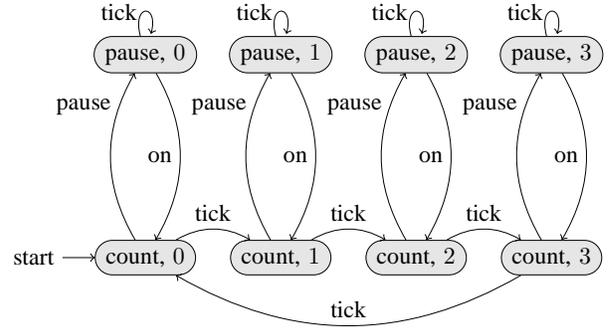
\begin{figure}[h]
\scalebox{0.9}{
  \begin{tikzpicture}
    \tikzstyle{every state}=[fill=gray!20!white,minimum size=0em,shape=rounded rectangle]
    
    \node[initial,state,fill=gray!20] (A) {count, $0$} ;
    \node[state, fill=gray!20] at (2,0) (B) {count, $1$} ;
    \node[state, fill=gray!20] at (4,0) (C) {count, $2$} ;
    \node[state, fill=gray!20] at (6,0) (D) {count, $3$} ;
    
    \node[state,fill=gray!20] at (0, 3)  (E) {pause, $0$} ;
    \node[state, fill=gray!20] at (2,3) (F) {pause, $1$} ;
    \node[state, fill=gray!20] at (4,3) (G) {pause, $2$} ;
    \node[state, fill=gray!20] at (6,3) (H) {pause, $3$} ;
    
    \path[->] (A) edge[bend left] node [above] {tick}  (B);
    \path[->] (A) edge[bend left] node [pos=0.8,left] {pause}  (E);
    \path[->] (E) edge[bend left] node [left] {on}  (A);
    \path[->] (E) edge[loop above] node [left] {tick}  (E);

    \path[->] (B) edge[bend left] node [above] {tick}  (C);
    \path[->] (B) edge[bend left] node [pos=0.8,left] {pause}  (F);
    \path[->] (F) edge[bend left] node [left] {on}  (B);
    \path[->] (F) edge[loop above] node [left] {tick}  (F);

    \path[->] (C) edge[bend left] node [above] {tick}  (D);
    \path[->] (C) edge[bend left] node [pos=0.8,left] {pause}  (G);
    \path[->] (G) edge[bend left] node [left] {on}  (C);
    \path[->] (G) edge[loop above] node [left] {tick}  (G);

    \path[->] (D) edge[bend left] node [above] {tick}  (A);
    \path[->] (D) edge[bend left] node [pos=0.8,left] {pause}  (H);
    \path[->] (H) edge[bend left] node [left] {on}  (D);
    \path[->] (H) edge[loop above] node [left] {tick}  (H);
  \end{tikzpicture}
}
\caption{State transition graph for a mod-4 counter.} 
\label{counter-sem}
\end{figure}
  We represent a state using a rounded rectangle and a transition using a
  labeled edge between participating states. 
  An initial state is marked using an incoming arrow to that state labeled
  ``start''. 
  An example of a run is the finite sequence: 
  \begin{eqnarray*}
    \seq{\text{(count, 0), tick, (count,1), pause, (pause, 1), tick,} \\
    \text{(pause, 1), on, (count, 1), tick, (count, 2)}}.
\end{eqnarray*}
\end{example}

A state transition graph is a feasible way to represent and computationally
analyze dynamical systems with finitely many states. 
However, to enable computational analysis of a general infinite state dynamical
system we need a finitary way to represent a potentially infinite space of states.
We begin this section by introducing concepts and notation used throughout this
article, followed by  discussing such syntactical models to represent
purely discrete and purely continuous dynamical system.
After introducing these models we present hybrid automata capable of modeling
hybrid dynamical systems.

\subsection*{Variables and Predicates}
Let $\Real$ be the set of real numbers, $\Rplus$ be the set of non-negative
real numbers, and $\Int$ be the set of integers. 

Let $X$ be a set of real-valued variables. 
A \emph{valuation} on $X$ is a function $\nu : X {\to} \Real$ and we write $V(X)$
for the set of valuations on $X$.
Abusing notation, we also treat a valuation $\nu$ as a point in $\Real^n$ that
is equipped with the standard \emph{Euclidean norm} $\norm{{\cdot}}$ where $n$ is
the cardinality of $X$.

%The \emph{predicates} are the most general form of constraints and 
We define a predicate over a set $X$ as a subset of $\Real^{|X|}$. 
For efficient computer-readable representation of predicates we often define
them using non-linear algebraic equations involving $X$. 
We write $\pred(X)$ for the set of predicates over $X$.
For a predicate $\pi \in \pred(X)$ we write $\sem{\pi}$ for the set of 
valuations in $\Real^{|X|}$ satisfying the equation $\pi$. 
We write $\top$ for the predicate that is true for all valuations, while $\bot$
for the predicate which is false for all the valuations. 

\begin{example}
An example of a predicate over the variables $\ddot{\theta}$ and $\theta$  is
\[ 
m\ell\ddot{\theta} = -m g \sin(\theta),
\]
characterizing the motion of an idealized pendulum (Figure~\ref{pendulum}) 
where $\theta$ is the angle the pendulum forms with its rest position,
$\ddot{\theta}$ is second derivative of $\theta$, $m$ is the mass of the
pendulum, $g$ is the gravitational constant, and $\ell$ is the length of the
pendulum.  
\end{example}

We say that a predicate $P$ is \emph{polyhedral} if it is defined as the conjunction
of a finite set of linear constraints of the form 
$ a_1 x_1 + \dots + a_n x_n \bowtie k$, where $k \in \Int$, for all $1 \leq i \leq
n$ we have that $a_i \in \Real, x_i \in X$, and $\bowtie \in \{<,\leq, =, >,
\geq\}$.  
An example of a polyhedral predicate over the set $\set{x, y, x}$ is
$2x+3y-9z\leq 5$.
We define an \emph{octagonal} predicate as the conjunction of a
finite set of linear constraints over $X$ of the form $\pm x \pm y \bowtie k \text{ or }
x \bowtie k$, where $k \in \Real$, $x, y \in X$. %, and $\bowtie \in \{<,\leq, =, >, \geq\}$.
Similarly a \emph{rectangular} predicate is defined as the conjunction of a
finite set of linear constraints over $X$ of the form $x \bowtie k$, where 
%$\bowtie \in \{<,\leq, =, >, \geq\}$, 
$k \in \Real$, and $x \in X$.
%Finally, a \emph{singular} predicate characterizes a \emph{singleton} set and it
%can be defined as the set of constraints specifying a fixed value for each
%variable, i.e. $x_i {=} k_i$ for $x_i {\in} X$ and $k_i {\in} \Real$.  
%We write $\poly(X)$, $\oct(X)$, $\rect(X)$, and $\sing(X)$ for the set of
%polyhedral, octagonal, rectangular, and singular predicates, resp., over the set $X$.

\subsection{Discrete Dynamical Systems}
Discrete dynamical systems can be conveniently modeled as \emph{extended finite state
  machines} having finitely many modes (or modes) and transitions between these
modes. 
The values of variables remain unchanged while the system is in some mode,
and changes only when a transition takes place where they can ``jump'' to new
values assigned by the transition.
These jumps are specified using \emph{predicates} over the set $X {\cup} X'$ that
relates the current values of variables of system, given as the set $X$, to the
values in the next time-step, given as the set $X'$ of primed-versions of variables
in $X$.
Transitions are often guarded by \emph{predicates} over variables specifying the
enabledness condition of the transition. 
Starting from some initial valuation to the variables, a system modeled using an
extended finite state machine evolves in discrete time-steps.
At each discrete step the system can take any enabled transition, i.e. satisfied
by the current variable valuation, and after executing the transition the
valuation of the variables is changed according to the jump condition.  
The system continues evolving in this fashion forever.
An extended finite state machine is formally defined as the following. 
\begin{definition}[Extended Finite State Machines: Syntax]
  An \emph{extended finite state machine} is a tuple 
  $\Mm = (M, M_0, \Sigma, X, \Delta, I, V_0)$ such that: 
  \begin{itemize}
  \item $M$ is a finite set of control \emph{modes} including a distinguished
    initial set of control modes $M_0 \subseteq M$, 
  \item $\Sigma$ is a finite set of \emph{actions},
  \item $X$ is a finite set of real-valued \emph{variable},
  \item $\Delta \subseteq M \times \pred(X) \times \Sigma \times \pred(X\cup X')
    \times M$ is the \emph{transition relation},  
  \item $I: M \to \pred(X)$  is the mode-invariant function, and
  \item $V_0 \in \pred(X)$ is the set of initial valuations.
  \end{itemize}
\end{definition}
For a transition $\delta= (m, g, a, j, m') \in \Delta$ we refer to $m \in M$ as
its \emph{source mode}, $g \in \pred(X)$ as its \emph{guard}, $a \in A$ as its
\emph{action}, $j \in \pred(X{\cup}X')$ as its \emph{jump constraint}, and
$m'\in M$ as the \emph{target mode}.    

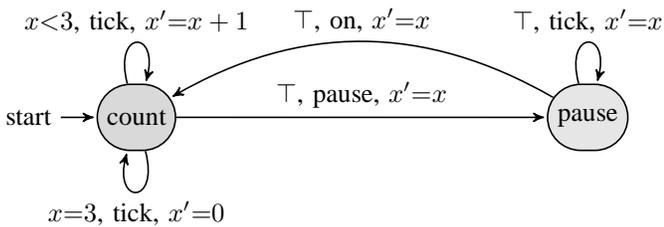
\begin{figure}[b]
\begin{center}
  \begin{tikzpicture}[->,>=stealth',shorten >=1pt,auto,node distance=2.8cm,
    semithick]
    \tikzstyle{every state}=[fill=gray!20!white,shape=rounded rectangle]
    
    \node[initial,state,fill=gray!30] (A) {count} ;
    
    \node[state,fill=gray!20] at (6,0) (B) {pause} ;
    
    \path (A) edge node [above]{$\top$, pause, $x'{=}x$} (B);
    \path (A) edge[loop above] node [above]{$x {<} 3$, tick, ${x'{=}x+1}$} (A);
    \path (A) edge[loop below] node [below]{$x {=} 3$, tick, ${x'{=}0}$} (A);
    \path (B) edge [bend right] node [above]{$\top$, on, $x'{=}x$} (A);
    \path (B) edge[loop above] node [above]{$\top$, tick, ${x'{=}x}$} (B);
  \end{tikzpicture}
  \caption{An EFSM description of a mod-4 counter with reset and pause.} 
  \label{counter}
\end{center}
\end{figure}
A configuration of an extended finite state machine is a tuple $(m, \nu)$ where $m$ is a control
mode and $\nu$ is a valuation of variables in $X$.
The execution of an  extended finite state machine begins in a configuration $(m_0, \nu_0)$ such that
the control mode $m_0 \in M_0$ is in the set of initial control modes and the
valuation $\nu_0 \in V_0$ satisfies the invariant of mode $m_0$, i.e. $\nu_0 \in
\sem{I(m_0)}$. 
At each discrete time-step the system executes a transition $(m, g, a, j, m')$
that is enabled in the current configuration $(m, \nu)$, i.e., $\nu \in
\sem{g}$, and the configuration of the system jumps to a new
configuration $(m', \nu')$ while respecting the jump constraints, i.e.  
$(\nu, \nu') \in \sem{j}$ as well as the invariant condition of the
resulting mode $\nu' \in \sem{I(m')}$.   
The system continues its execution from the resulting configuration in the
similar fashion. 
Hence, we can define the semantics of an extended finite state machine as a
state transition graph in the following manner. 
\begin{definition}[Extended Finite State Machine: Semantics]
  The semantics of an extended finite state machine $\Mm = (M, M_0, \Sigma, X,
  \Delta, I, V_0)$ is given as a state transition graph $T^\Mm = (S^\Mm,
  S_0^\Mm, \Sigma^\Mm, \Delta^\Mm)$ where:  
  \begin{itemize}
  \item $S^\Mm \subseteq (M \times \Real^{|X|})$ is the set of configurations of
    $\Mm$ such that for all $(m, \nu) \in S^\Mm$ we have that $\nu \in
    \sem{I(m)}$; 
  \item $S_0^\Mm \subseteq S^\Mm$ such that $(m, \nu) \in S^\Mm$ if $m \in M_0$
    and $\nu \in V_0$;
  \item $\Sigma^\Mm = \Sigma$ is the set of labels;
  \item $\Delta^\Mm \subseteq S^\Mm \times \Sigma^\Mm \times S^\Mm$ is the set of
    transitions such that $((m, \nu), a, (m', \nu')) \in \Delta^\Mm$ if there
    exists a transition $\delta = (m, g, a, j, m') \in \Delta$ such that
    the current valuation $\nu$ satisfies the guard of $\delta$, i.e. $\nu \in
    \sem{g}$; 
    the pair of current and next valuations $(\nu, \nu')$ satisfies the jump
    constraint of $\delta$, i.e. $(\nu, \nu') \in \sem{j}$; and 
    the next valuation satisfies the invariant of the target mode of $\delta$,
    i.e. $\nu' \in \sem{I(m')}$. 
  \end{itemize}
\end{definition}

Let us consider an example of the syntax and semantics of an extended finite
state machine. 
\begin{example}[Modulo-4 counter]
  Let us consider a modulo-4 counter with reset and pause functionality shown in
  Figure~\ref{counter}.
  This extended finite state machine $\Mm = (M, M_0, \Sigma, X, \Delta, I, V_0)$ has two control modes
  $M = \set{\text{count}, \text{pause}}$ with count being the initial mode.
  The variable $x$ is the only variable, while the set of action is $\Sigma =
  \set{tick, on, pause}$ where tick, on, and pause stand for clock-tick,
  start-counting, and pause-counting actions, respectively.
  While drawing an  extended finite state machine, we depict modes by rounded rectangles and transitions
  by arrows connecting the modes labeled by a triplet $(g, a, j)$
  showing the guard, the action, and the jump predicate of the transition. 
  For example the transition  $(\text{count}, x=3, t, x'=0, \text{count})$ is
  shown in the Figure~\ref{counter} as a self-loop labeled with $(x=1,t, x'=0)$
  on the mode labeled count. 
  It is straightforward to see that the  extended finite state machine in Figure~\ref{counter} models a
  modulo-4 counter with reset and pause.
  The corresponding state transition graph is shown in the
  Figure~\ref{counter-sem}. 
\end{example}
In the rest of the article, to minimize clutter, we will omit the guard if it is
the predicate $\top$, and we  omit the jump predicates specifying that the value
of a variable remains unchanged, i.e. predicates of the form $x' = x$.

\subsection{Continuous Dynamical Systems}
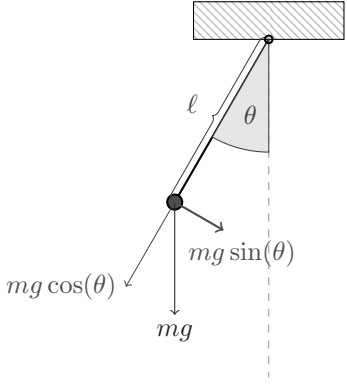
\begin{figure}[t]
  \begin{center}
    \begin{tikzpicture}[scale=0.5]
  
  \draw[color=black!40, dashed] (0,0)--(0,-9);

  \pgftransformrotate{-30}
  \draw[color=black, thick] (0,0)--(0,-5);

  \pgftransformrotate{30}
  \filldraw[pattern=north west lines, pattern color=gray!60] (-2, 0) rectangle (2,1);
  
  \filldraw[fill=gray!20,draw=gray!50!black] (0,0) -- (0,-3) arc (-90:-120:3) -- cycle;
  \node[color=gray!50!black] at (-.49,-2) {$\theta$};
  \pgftransformrotate{-30}
  \draw[color=black,fill=black!70, thick] (0,-5) circle (.2cm);
  \pgftransformrotate{30}
  \draw[->, color=gray!30!black] (-2.5,-4.33)--(-2.5,-7.33)
  node[below] {$mg$};
  \pgftransformrotate{-30}
  \draw[->,color=black!70, thick] (0,-5)--(1.5,-5) node[below] {$~~~~mg\sin(\theta)$};
  \draw[->,color=black!70] (0,-5)--(0,-7.6) node[left] {$mg\cos(\theta)$};
  \draw[snake=brace, color=gray!50!black] (-.1,-5)--(-.1,0);
  \node[color=gray!50!black] at (-.9,-2.5) {$\ell$};
  \draw[color=black, thick] (0,0) circle (.1cm);
\end{tikzpicture}
\caption{An idealized pendulum with length $\ell$ and mass $m$.} 
\label{pendulum}
\end{center}
\end{figure}
 
For the purpose of this article, a continuous dynamical system is a finite set of
continuous variables along with a set of ordinary differential equations
characterizing the dynamics or the flow of these variables as a function of
time. 
We represent the flow of a continuous dynamical system using a flow function 
$F: \Real^{|X|} \to \Real^{|X|}$ characterizing the system of ordinary
differential equations:
\begin{equation}
\dot{X} = F(X)~\label{ode1}
\end{equation}
where, following Newton's dot notation for differentiation, $\dot{X}$ represents
the set of first-order derivatives of the variables in the set $X$.
Information about the higher-order derivatives can be represented using only
first-order derivatives introducing auxiliary variables. 
For example the second-order differential equation $\ddot{\theta} +
(g/\ell)\sin(\theta) = 0$  can be written as a system of first-order
differential equations $\dot{\theta} = y, \dot{y} = -(g/\ell)\sin(\theta)$. 
Formally, a continuous dynamical system is defined in the following manner. 

\begin{definition}[Continuous Dynamical System] 
  A \emph{continuous dynamical system} is a tuple $\Mm = (X, F, \nu_0)$ such that
  \begin{itemize}
  \item $X$ is a finite set of real-valued \emph{variable}, 
  \item $F: \Real^{|X|} \to \Real^{|X|}$ is the flow function characterizing the
    the set of ordinary differential equation $\dot{X} = F(X)$, and
  \item $\nu_0 \in \Real^{|X|}$ is the initial valuation.
  \end{itemize}
\end{definition}

A \emph{run} of a continuous dynamical system $\Mm {=} (X,
F, \nu_0)$ is given as a solution to the system of differential equations
(\ref{ode1}) with initial valuation $\nu_0$.
Let  a differentiable function ${f: \Rplus {\to} \Real^{|X|}}$ be a
solution to (\ref{ode1}), that provides the valuations of the variables as a
function of time, such that:
\begin{eqnarray*}
  f(0) &=& \nu_0 \\
  \dot{f}(t) &=& F(f(t)) \text{ for every $t \in \Rplus$,}
\end{eqnarray*}
where $\dot{f}: \Rplus {\to} \Real{|X|}$ is the time derivative of the function
$f$.  
We call such a function $f$ a run of the continuous dynamical system $\Mm$.
Since, in general, a solution of (\ref{ode1}) may not exist or may not be unique,
a run of a continuous dynamical system may not exist or may not be
unique~\cite{LTS08}.   
To ensure the existence and the uniqueness of the run we enforce
Lipschitz-continuity\footnote{We say that a function ${F: \Real^n \to
    \Real^n}$ is Lipschitz-continuous if there exists a constant $K{ >} 0$,
  called the Lipschitz constant, such that for all $x, y \in \Real^n$ we have
  that $\|F(x) - F(y)\| < K \|x-y\|$.}  assumption on $F$.
The following result states the existence and uniqueness of the set of
ordinary differential equations (\ref{ode1}) under Lipschitz-continuity
  assumption. 
\begin{theorem}[Picard-Lindel\"of Theorem~\cite{Roy09}]
  If a function $F: \Real^{|X|} \to \Real^{|X|}$ is Lipschitz-continuous then the
  differential equation $\dot{X} {=} F(X)$ with initial valuation $\nu_0 {\in}
  \Real^{|X|}$ has a unique solution $f: \Rplus {\to} \Real^{|X|}$ for all  $\nu_0
  {\in} \Real^{|X|}$. 
\end{theorem}

In addition, Lipschitz-continuity offers the following advantage while numerically
simulating an approximate solution to the differential equations~(\ref{ode1}).

\begin{theorem}[Stability wrt initial valuation~\cite{LTS08}]
  Let $F$ be a Lipschitz-continuous function with constant $K {>}0$ and let $f{:}
  \Rplus {\to} \Real^{|X|}$  and $f'{:}\Rplus {\to} \Real^{|X|}$ be solutions to the
  differential equation $\dot{X} {=} F(X)$ with initial valuation $\nu_0 {\in}
  \Real^{|X|}$ and $\nu_0' {\in} \Real^{|X|}$, respectively.
  Then, for all $t {\in} \Rplus$ we have that 
  $\|f(t) {-} f'(t)\| \leq \|\nu{-}\nu_0\|e^{Kt}$.
\end{theorem}

This theorem implies that, under Lipschitz-continuous assumption on the flow
function $F$, any two runs whose initial valuation is close to one-another
remain close as the time progresses. 
Since it is not always possible to analytically solve differential equations,
this property permits us to numerically simulate the behaviour of continuous
dynamical system using approximation methods, e.g. Euler's method or
Runge-Kutta method, that are readily available in tools such as
Matlab~\cite{Matlab} and Mathematica~\cite{Mathematica}.   
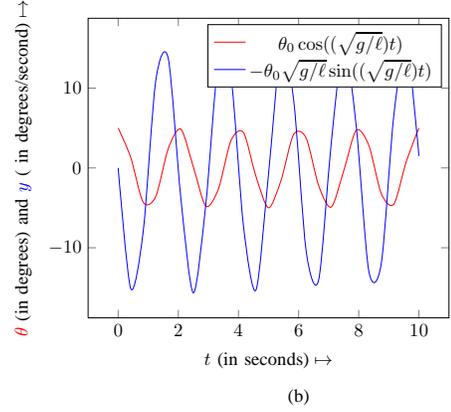
\begin{figure}
\begin{center}
  \scalebox{0.7}{
\begin{tikzpicture}
  \begin{axis}[domain=0:10, 
    xlabel=$t \text{ (in seconds)} \mapsto $,
    ylabel=${{\color{red}\theta} \text{ (in degrees)} \text{ and }
      {\color{blue} y} \text{ ( in degrees/second)}} \mapsto$
    ]
     \addplot[smooth,red] {((5)*cos(deg(sqrt(9.81)*x)))}; 
      \addplot[smooth,blue] {(-(sqrt(9.81)*(5)*sin(deg(sqrt(9.81)*x))))}; 
      \legend{$\theta_0\cos((\sqrt{g/\ell})t)$,$
        -\theta_0\sqrt{g/\ell}\sin((\sqrt{g/\ell})t)$};
    \end{axis}
    \node at (4, -1.5) {(b)};
\end{tikzpicture}
}
\caption{The variables $\theta$ (angle displacement) and $y$ (angular velocity) are
  plotted with respect to the time for a pendulum with $\ell=1$ meter with
  $\theta_0 = 5$ degrees.} 
\label{pendulum-phase}
\end{center}
\end{figure}
\begin{example}[Simple Pendulum]
  Consider a simple pendulum shown in Figure~\ref{pendulum} and its the
  motion equations: 
  \begin{eqnarray*}
    \dot{\theta} &=& y,\\
    \dot{y} &=& -(g/\ell)\sin(\theta),
  \end{eqnarray*}
  with initial valuations $(\theta, y) = (\theta_0, 0)$.
  To analytically solve these equations let us assume small enough angular
  displacement $\theta$ and $\sin(\theta) \approx \theta$.
  Now the equations simplify to  
  \[
  \dot{\theta} = y \text{ and } \dot{y} = -(g/\ell)\theta. 
  \]
  Hence our continuous dynamical system is $\Mm = (X, F, \nu_0)$ where $X  =
  \set {\theta, y}$, $F$ is  such that $F(\dot{\theta}) = y$ and $F(\dot{y}) =
  -(g/\ell)\theta$ and $\nu_0 = (\theta_0, 0)$.
  The solution for these differential equations is 
  \begin{eqnarray*}
    \theta(t) &=& A \cos (Kt) + B \sin (Kt)\\
    y(t) &=& - AK  \sin (Kt) + BK \cos (Kt),\\
  \end{eqnarray*}
  where $K = \sqrt{g/\ell}$.
  Substituting $\theta(0) = \theta_0$ and $y(0)= 0$ from the initial valuation,
  we get that $A = \theta_0$ and $B = 0$.
  Hence the unique run of the pendulum system can be given as the function $f:
  \Rplus \to \set{\theta, y}$ as $ t \mapsto (\theta_0 \cos(Kt), - \theta_0 K \sin(Kt))$.
  Figure~\ref{pendulum-phase} shows the change in valuations of the variables
  $\theta$ and $y$ as a function of time. 
\end{example}

\subsection{Hybrid Dynamical Systems}
\begin{figure*}[t]
  \begin{center}
   \begin{tikzpicture}[scale=0.69]
      \begin{axis}[
        xlabel=$t$,
        ylabel=$X$]       ]
        \addplot[dashed,mark=*,blue] plot coordinates 
        {
          (0,2)
          (1, 3)
          (2, 4)
          (3, 3)
          (4, 2)
          (5, 1)
        };
        \addlegendentry{Discrete System}
      \end{axis}
    \end{tikzpicture}
    \begin{tikzpicture}[scale=0.69]
      \begin{axis}[
        domain=0:10,
        xlabel=$t$,
        ylabel=$X$]
        \addplot[smooth,red!30]{sin(deg(x))};
        \addlegendentry{Continuous System}
      \end{axis}
    \end{tikzpicture}
    \begin{tikzpicture}[scale=0.69]
      \begin{axis}[
        domain=0:10,
        xlabel=$t$,
        ylabel=$X$]
        \addplot[domain=0:2,smooth,red!30]{sin(deg(x))};
        \addplot[dashed,mark=*,blue] plot coordinates 
        {
          (2,0.9)
          (2, 0)
        };
        \addplot[domain=2:3,smooth,red!30]{3*x-6};%cos(deg(x))+0.44};
        \addplot[dashed,mark=*,blue] plot coordinates 
        {
          (3, 3)
          (3, -0.8)
        };
        \addplot[domain=3:6,smooth,red!30]{cos(deg(x^2))};%cos(deg(x))+0.44};
        \addplot[dashed,mark=*,blue] plot coordinates 
        {
          (6, 0)
          (6, 1)
        };
        \addplot[domain=6:10,smooth,red!30]{cos(deg(x))};%cos(deg(x))+0.44};
        \addlegendentry{Hybrid System}
      \end{axis}
    \end{tikzpicture}
  \end{center}
  \caption{Runs of discrete, continuous, and hybrid systems.}
  \label{flowjump}
\end{figure*}
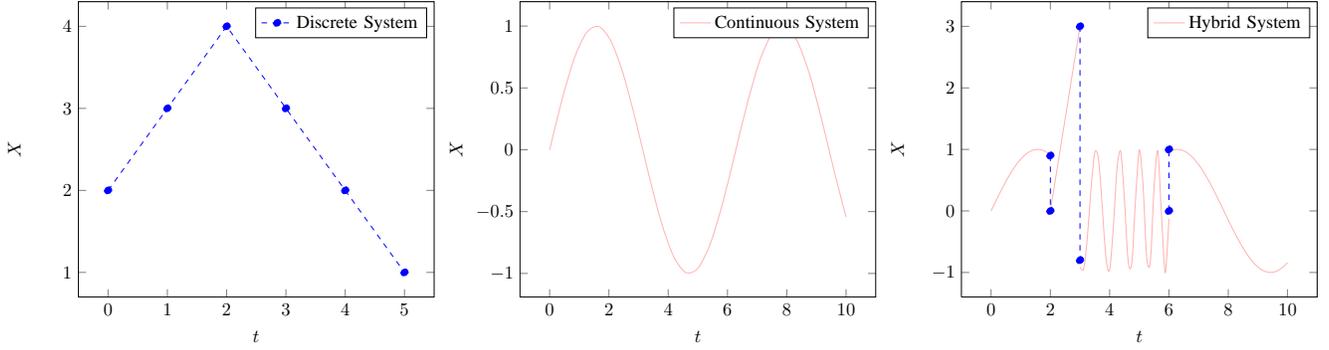
In the previous two subsections we discussed modeling of purely discrete and
purely continuous dynamical systems.   
We saw that in a discrete dynamical system the state of the system changes
during a discrete transition where it ``jumps'' (see Figure~\ref{flowjump}) to
the new value governed by the transition relation, while in a continuous system
the state of the system continuously ``flows'' (see Figure~\ref{flowjump}) in a
fashion governed by ordinary differential equations.  
Hybrid systems share their properties with both discrete as well as continuous
systems, as their state progresses with time in both discrete jumps as well as
continuous flows.
In this section we present hybrid automata, a combination of extended finite
state machines and continuous dynamical systems, where in every control mode the
dynamics of the variables of the system can be specified using ordinary
differential equations.  

\begin{definition}[Hybrid Automata: Syntax]
  A hybrid automaton is a tuple $\Hh= (M, M_0, \Sigma, X, \Delta, I, F, V_0)$
  where: 
  \begin{itemize}
  \item $M$ is a finite set of control \emph{modes} including a distinguished
    initial set of control modes $M_0 \subseteq M$, 
  \item $\Sigma$ is a finite set of \emph{actions},
  \item $X$ is a finite set of real-valued \emph{variable},
  \item $\Delta \subseteq M \times \pred(X) \times \Sigma \times \pred(X\cup X')
    \times M$ is the \emph{transition relation},   
  \item $I: M \to \pred(X)$  is the mode-invariant function,
  \item $F: M \to (\Real^{|X|} \to \Real^{|X|})$ is the mode-dependent flow
    function characterizing the flow for each mode $m \in M$ as the set of
    ODEs $\dot{X} = F(m)(X)$, and  
  \item $V_0 \in \pred(X)$ is the set of initial valuations.
  \end{itemize}
  To ensure existence of unique solutions of the ODEs in flow functions, we
  assume that for each mode $m \in M$ the flow function $F(m)$ is
  Lipschitz-continuous.   
\end{definition}

Just like in an  extended finite state machine, a configuration of a hybrid automaton is a tuple $(m,
\nu)$ where $m \in M$ is a mode and $\nu \in \Real^{|X|}$ is a variable
valuation.  
For a Lipschitz-continuous flow function $F: M \to (\Real^{X} \to \Real^{|X|})$,
a valuation $\nu \in \Real^{|X|}$, a mode $m \in M$, and a time delay $t \in
\Rplus$ we define $(\nu {\pplus{F(m)}} t)$ for the unique valuation $f(t)$ where
$f$ is the unique run of the continuous dynamical system $(X, F(m), \nu)$.  
For a jump predicate $j \in \pred(X\cup X')$ and valuation $\nu$ we define
$\nu[j]$ for the set of valuations $\nu' \in \Rplus^{|X|}$ such that $(\nu,
\nu') \in j$. 

The execution of a hybrid automaton begins in an initial configuration 
$(m_0, \nu_0)$ where $m_0 \in M_0$ is an initial mode and $\nu_0 \in V_0$ is an
initial valuation satisfying $\nu_0 \in \sem{I(m_0)}$.
The system stays in a mode for some time, say $t_1 \in \Rplus$, and
while the system stays in a control mode $m$ the valuation of the variables
changes according to ODE specified by the flow $F(m)$ of the corresponding mode. 
After spending $t_1 \in \Rplus$ time in mode $m_0$ an enabled transition $(m_0, g,
a, j, m_1)$ is nondeterministically chosen and executed. 
Notice that we say that a transition $(m_0, g, a, j, m_1)$ is enabled if $(\nu_0
\pplus{F(m_0)} t_1) \in \sem{g}$ and all the intermediate valuations that system
passes through from $\nu_0$ to $(\nu_0 \pplus{F(m_0)} t_1)$ satisfy the invariant
of the mode $m_0$, i.e. for all $t \in [0, t_1]$ we have that $(\nu_0
\pplus{F(m_0)} t) \in \sem{I(m_0)}$.
After executing the transition $(m_0, g, a, j, m_1)$ the state  of the
system jumps to a new configuration $(m_1, \nu_1)$ such that $\nu_1 \in
\sem{I(m_1)}$ and $\nu_1 \in (\nu_0 \pplus{F(m_0)} t_1)[j]$.
The system continues its operation in a similar manner from the resulting
configuration $(m_1, \nu_1)$. We can formalize this semantics using a
(uncountably infinite) state transition graph. 
\begin{definition}[Hybrid Automata: Semantics]
\label{ha:sem}
The semantics of a hybrid automaton $\Hh {=} (M, M_0, \Sigma, X, \Delta, I, F, V_0)$ is
  given as a state transition graph $T^\Hh {=} (S^\Hh, S_0^\Hh, \Sigma^\Hh,
  \Delta^\Hh)$ where:  
  \begin{itemize}
  \item $S^\Hh \subseteq (M \times \Real^{|X|})$ is the set of configurations of
    $\Hh$ such that for all $(m, \nu) \in S^\Hh$ we have that $\nu \in
    \sem{I(m)}$; 
  \item $S_0^\Hh \subseteq S^\Hh$ s.t. $(m, \nu) \in S^\Hh_0$ if $m \in
    M_0$ and $\nu \in V_0$;
  \item $\Sigma^\Hh = \Rplus \times \Sigma$ is the set of labels;
  \item $\Delta^\Hh \subseteq S^\Hh \times \Sigma^\Hh \times S^\Hh$ is the set of
    transitions such that $((m, \nu), (t, a), (m', \nu')) \in \Delta^\Hh$ if 
    there exists a transition $\delta = (m, g, a, j, m') \in \Delta$ such that 
    \begin{itemize}
    \item 
      $(\nu {\pplus{F(m)}} t) \in \sem{g}$;
    \item
      $(\nu \pplus{F(m)} \tau) \in \sem{I(m)}$ for all $\tau \in [0, t]$;
    \item 
      $\nu' \in (\nu \pplus{F(m)} t)[j]$; and
    \item
    $\nu' \in \sem{I(m')}$. 
    \end{itemize}
  \end{itemize}
\end{definition}

\begin{example}[A bouncing ball] 
  In Figure~\ref{ball} we model a bouncing ball using a hybrid automaton with
  one control mode $m$ and two variables: the variable $x_1$, representing the
  vertical position of the ball, and the variable $x_2$, representing the vertical
  velocity of the ball. 
  \begin{figure}[h]
  \begin{center}
    \begin{tikzpicture}
      \tikzstyle{every state}=[fill=gray!20!white,shape=rounded rectangle]
      
      \node[initial,state,fill=gray!30,label=below:$x_1\geq0$] (A) {$\begin{array}{l}\dot{x_1}=x_2,\\
          \dot{x_2} = -g\\ m  \end{array}$} ;
      
      \path[->,thick] (A) edge[loop above] node [above]{$\begin{array}{ll}
          x_1{=}0 \wedge x_2 {\leq} 0,\\ \text{ impact } \\ x_1'{ =}
          x_1 \wedge x_2'{=} -c x_2 \end{array}$} (A);
    \end{tikzpicture}
  \end{center}
  \caption{ A hybrid automaton modeling the dynamics of a bouncing ball}
  \label{ball}
\end{figure}
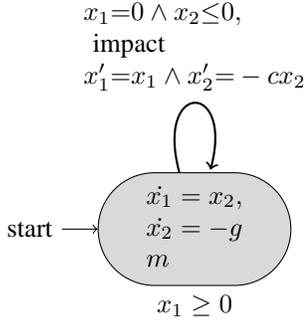 

  The differential equations governing the free fall of the ball can be given
  using Newton's law of motion as $\dot{x_1} = x_2$ and $\dot{x_2} = -g$.
  The valuations of the variables flow according to these equations until the ball
  comes in the contact with ground, and at that time it reverses the direction
  of its velocity, while losing some energy proportional to its restitution
  coefficient $c$, i.e. after the impact we have $x_1' = x_1$ and $x_2'=-c x_2$.
% Notice that the guard of the transition is $x_1{=}0$ and $x_1\leq0$ requiring
%  that the vertical position is $0$, and the new valuation of the variables is
%  given as $x_1' = x_1$ and $x_2'=-c x_2$.
  Observe that the bouncing ball system is a hybrid system since its dynamics
  involve both flows and jumps. 
  The continuous dynamics of the system is captured using flow function of the
  unique mode $m$, while the jump is modeled with the discrete transition
  labeled \emph{impact}. 
  For the starting valuation we assume $x_1 = \ell$ meters and $x_2 = 0$.
  Formally the hybrid automata $H = (M, M_0, \Sigma, X, \Delta, I, F, V_0)$
  models the bouncing ball where: 
  \begin{itemize}
  \item $M = M_0 = \set{m_0}$,
  \item $\Sigma = \set{\text{impact}}$,
  \item $X = \set{x_1, x_2}$,
  \item $\Delta$ contains the following transition 
    \[ 
    (m, x_1{=}0 \wedge x_2 {\leq} 0, \text{impact}, x_1'{ =}
    x_1 \wedge x_2'{=} -c x_2, m),
    \]
  \item $I(m) = x_1 {\geq} 0$,
  \item $F(m) = \dot{x_1}=x_2 \wedge \dot{x_2} = -g$, and
  \item $V_0 = \set{(\ell, 0)}$.
  \end{itemize}
%  \begin{example}{The Bouncing Ball}
  The transition diagram corresponding to this automaton is shown in
  Figure~\ref{ball}(a). 
  The transition diagram of a hybrid automaton follows the similar conventions as
  that of an  extended finite state machine, with the exception of flow conditions.
  We write flow conditions of a mode inside the rounded rectangle representing
  the mode.  

  Now let us explain the unique run of the system starting from the configuration
  $(m, (\ell, 0))$.
%  For the sake of convenience we assume $c = 1$.
  The solution to ODE corresponding to the flow function is 
  \begin{equation}
    \label{eq1}
  x_1(t) = - \frac{1}{2} g t^2 + C t + D \text{ and } 
  x_2(t) = - g t + C.
  \end{equation}
  For the initial configuration is $(m, (\ell, 0))$ solving (\ref{eq1}) we get $C
  = 0$ and $D=\ell$.
  Hence from $(m, (\ell, 0))$ system flows according to the equations $ x_1(t) = -
  \frac{1}{2} g t^2 + \ell$ and $x_2(t) = - g t$.
  According to these equations the value of variable $x_1$ continue to fall for
  the next $t_1 = \sqrt{2\ell/g}$ time units when $x_1$ becomes $0$, and the
  transition \emph{impact} becomes available and must be taken (since the
  invariant of the mode requires $x_1$ to be non-negative).
  Immediately before taking the transition the configuration is $(0, -g t_1)$.
  Using our notations we can write it as $(0, -g t_1) = (\ell, 0)\pplus{F(m)} t_1$.
  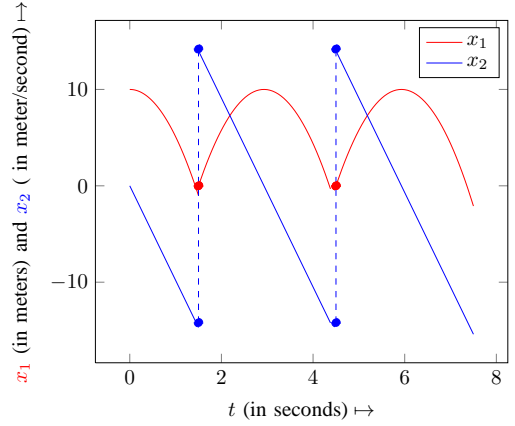
\begin{figure}
\begin{center}
\scalebox{0.8}{  \begin{tikzpicture}
      \begin{axis}[domain=0:10,
        xlabel=$t \text{ (in seconds)} \mapsto $,
        ylabel=${{\color{red}x_1} \text{ (in meters)} \text{ and }
          {\color{blue} x_2} \text{ ( in meter/second)}} \mapsto$
        ]
        \addplot[domain=0:1.5,smooth,red] {-9.8*0.5*x^2+10}; 
        \addplot[domain=0:1.5,smooth,blue] {-9.8*x}; 
        \addplot[dashed,mark=*,blue] plot coordinates 
        {
          (1.5,-14.2)
          (1.5, +14.2)
        };
        \addplot[dashed,mark=*,red] plot coordinates 
        {
          (1.5,0)
          (1.5, 0)
        };
        \addplot[domain=1.5:4.38,smooth,red] {-9.8*0.5*(x-1.5)^2+14*(x-1.5)}; 
        \addplot[domain=1.5:4.38,smooth,blue] {-9.8*(x-1.5)+14}; 
        \addplot[dashed,mark=*,blue] plot coordinates 
        {
          (4.5,-14.2)
          (4.5, +14.2)
        };
        \addplot[dashed,mark=*,red] plot coordinates 
        {
          (4.5,0)
          (4.5, 0)
        };
        \addplot[domain=4.5:7.5,smooth,red] {-9.8*0.5*(x-4.5)^2+14*(x-4.5)}; 
        \addplot[domain=4.5:7.5,smooth,blue] {-9.8*(x-4.5)+14}; 
        \legend{$x_1$, $x_2$};
      \end{axis}
    \end{tikzpicture}
}
\end{center}
  \caption{a run of the system where the initial vertical position is $\ell = 10$
    meters and the coefficient of restitution $c = 1$.} 
  \label{ball-phase}
\end{figure}
  
  After taking the transition \emph{impact}  this valuation
  changes according to the jump function $x_1'{ =} x_1 \wedge x_2'{=} -c x_2$
  resulting in the new valuation $(0, cgt_1)$. 
  Again, in our notation we write 
  $(0, cgt_1) \in (0, -gt_1)[x_1'{ =} x_1 {\wedge} x_2'{=} {-}c x_2]$.
  The run of the system, so far, can be written as $\seq{(m, (\ell, 0)),
    (t_1,\text{impact}), (m, (0, cgt_1))}$. 
  Now from the configuration $(m, (0, cgt_1))$ the system can flow continuously
  according to $F(m)$.  
  Solving (\ref{eq1}) for this initial valuation we get $C = cgt_1$ and $D =0$. 
  Hence from $(m, (0, cgt_1))$ the system flows according to the
  equations $x_1(t) = - \frac{1}{2} g t^2 + cgt_1 t$ and $x_2(t) = - g t
  +cgt_1$ for the next $t_2 = 2ct_1$ time units till it reaches
  the valuation $x_1 = 0$ (the ball hits the ground again).  
  At this point the resulting configuration will be $(0, -cgt_1)$ and after the
  transition the configuration will be $(0, c^2gt_1)$.
  The system continues in this fashion forever and realizes the following
  infinite run of the system:
  \begin{eqnarray}
    \langle
    (m, (\ell, 0)), &(t_1,\text{impact}),& (m, (0, cgt_1)),\nonumber\\
    & (2ct_1,\text{impact}),& (m, (0, c^2gt_1)), \nonumber\\ 
    & (2c^2t_1,\text{impact}),& (m, (0, c^3gt_1)), \ldots \rangle, \label{eqrun}
  \end{eqnarray}
  where $t_1 = \sqrt{2\ell/g}$.
  The first two transitions of the run for $\ell = 10$ and $c =1$ are shown in
  Figure~\ref{ball}(b).  
\end{example}

  \begin{figure*}[t]
  \begin{center}
   \scalebox{0.7}{ 
    \begin{tikzpicture}[->,>=stealth',shorten >=1pt,auto,semithick]
      \tikzstyle{every state}=[fill=gray!20!white,minimum size=3em,rounded rectangle]
      
      \node[initial,initial where=above,state,fill=gray!30] at (0, 0) (A) {$U_1$} ;
      \node[state,fill=gray!20] at (0, -2) (B) {$\begin{array}{lll} \dot{x_1} =
          1 \\ S_1 \end{array}$} ;
      \node[state,fill=gray!20] at (0, -5) (C) {$F_1$} ;
      
      \path (A) edge node {$\beg_1$} (B);
      \path (B) edge node {$\begin{array}{lll} x_1 = 3,\\  \fin_1, \\ \done_1' = 1\end{array}$} (C);
      
      \node at (0, -6) {$\Hh_{j_1}$};
    \end{tikzpicture}
  }
   \scalebox{0.7}{
     \begin{tikzpicture}[->,>=stealth',shorten >=1pt,auto,semithick]
      \tikzstyle{every state}=[fill=gray!20!white,minimum size=3em,rounded rectangle]
      
      \node[initial,initial where=above,state,fill=gray!30] at (0, 0) (A) {$U_2$} ;
      \node[state,fill=gray!20] at (0, -2) (B) {$\begin{array}{lll} \dot{x_2} =
          1 \\ S_2 \end{array}$} ;
      \node[state,fill=gray!20] at (0, -5) (C) {$F_2$} ;
      
      \path (A) edge node {$\done_1=1, \beg_2$} (B);
      \path (B) edge node {$\begin{array}{lll}x_2 = 4, \\ \fin_2,\\ \done_2' = 1\end{array}$} (C);
      
      \node at (0, -6) {$\Hh_{j_2}$};
    \end{tikzpicture}
    }
    \scalebox{0.7}{
    \begin{tikzpicture}[->,>=stealth',shorten >=1pt,auto,semithick]
      \tikzstyle{every state}=[fill=gray!20!white,minimum size=3em,rounded rectangle]
      
      \node[initial,initial where=above,state,fill=gray!30] at (0, 0) (A) {$I_1$} ;
      \node[state,fill=gray!20] at (-1.5, -5) (B) {$P_{1,1}$} ;
      \node[state,fill=gray!20] at (1.5, -5) (C) {$P_{1,2}$} ;
      
      \path (A) edge node[pos=0.8] {$\beg_1$} (B);
      \path (B) edge[bend left] node[pos=0.8] {$\fin_1$} (A);

      \path (A) edge node[left,pos=0.9] {$\beg_2$} (C);
      \path (C) edge[bend right] node[right,pos=0.9] {$\fin_2$} (A);

      \node at (0, -6) {$\Hh_{m_1}$};
    \end{tikzpicture}    
  }
    \scalebox{0.7}{
    \begin{tikzpicture}[->,>=stealth',shorten >=1pt,auto,semithick]
      \tikzstyle{every state}=[fill=gray!20!white,minimum size=3em,rounded rectangle]
      
      \node[initial,initial where=above,state,fill=gray!30] at (0, 0) (A)
      {$(U_1, U_2, I_1)$} ;

      \node[state,fill=gray!20] at (-2, -2) (B) {$\begin{array}{lll} \dot{x_1}
          =1 \\ (S_1, U_2, P_{1,1})\end{array}$} ;
      \node[state,fill=gray!20] at (2, -2) (C) {$\begin{array}{lll} \dot{x_2}
          =1 \\ (U_1, S_2, P_{1,2})\end{array}$} ;
      
      \node[state,fill=gray!20] at (-2, -4) (D) {$(F_1, U_2, I_1)$} ;
      \node[state,fill=gray!20] at (2, -4) (E) {$(U_1, F_2, I_1)$} ;

      \node[state,fill=gray!20] at (-2, -6) (F) {$\begin{array}{lll} \dot{x_2}
          =1 \\ (F_1, S_2, P_{1,2})\end{array}$} ;
      \node[state,fill=gray!20] at (2, -6) (G) {$\begin{array}{lll} \dot{x_1}
          =1 \\ (S_1, F_2, P_{2,1})\end{array}$} ;
      
      \node[state,fill=gray!20] at (0, -8) (H) {$(F_1, F_2, I_1)$} ;

      \path (A) edge node[left] {$\beg_1$} (B);
      \path (A) edge node {$\done_1=1,\beg_2$} (C);

      \path (B) edge node[left] {$x_1=3, \fin_1, \done_1'=1$} (D);
      \path (C) edge node {$x_2=4, \fin_2, \done_2'=1$} (E);

      \path (D) edge node[left] {$\done_1=1, \beg_2$} (F);
      \path (E) edge node {$\beg_1$} (G);

      \path (F) edge node[left] {$x_2=4, \fin_2, \done_2'=1$} (H);
      \path (G) edge node {$x_1=3, \fin_1, \done_2'=1$} (H);
     
      \node at (0, -9) {$\Hh_{j_1}\otimes\Hh_{j_2}\otimes\Hh_{m_1}$};
    \end{tikzpicture}    
}
\caption{Network of hybrid automata $\Hh_{j_1}$, $\Hh_{j_1}$, and $\Hh_{m_1}$
  corresponding to jobs $j_1$ and $j_2$, and a machine $m_1$, and their
  product automata $\Hh_{j_1}\otimes\Hh_{j_2}\otimes\Hh_{m_1}$.}
\label{job2}
\end{center}
\end{figure*}
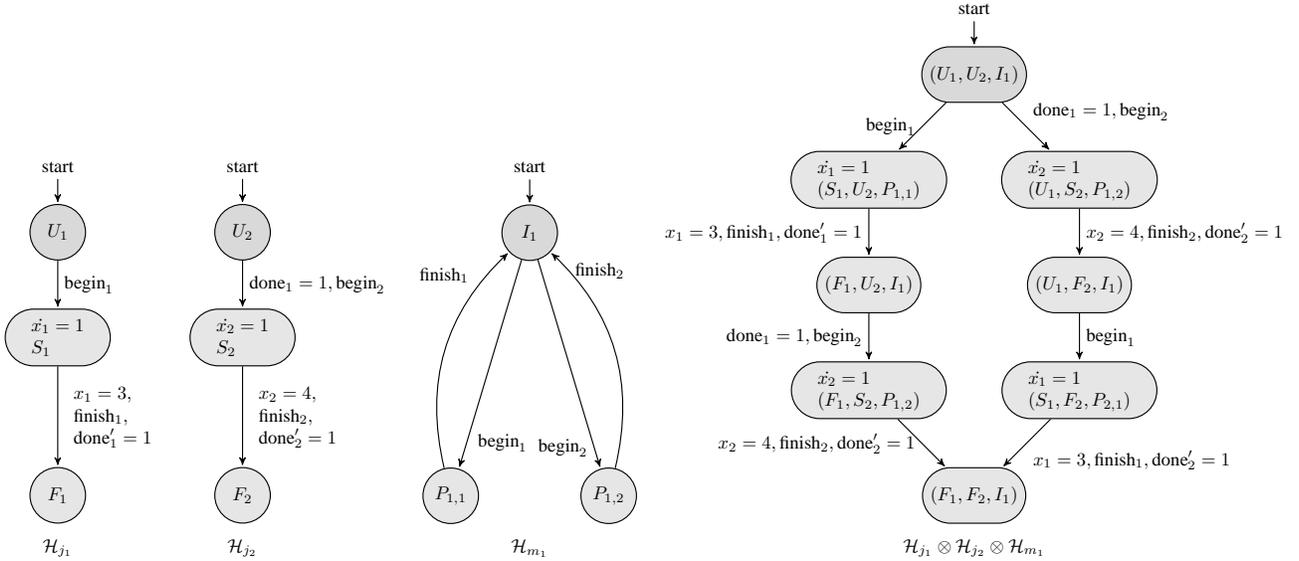
For a given run $r = \seq{(m_0, \nu_0), (t_1, a_1), (m_1, \nu_1),
  \ldots}$ of a hybrid automaton we define its time $\tot{r}$ is defined as 
\[
\tot{r} =  \sum_{i=1}^{\infty} t_i. 
\] 
We say that a run $r$ \emph{time-diverging} if $\tot{r} = \infty$. 
For an example of a time-diverging run consider~(\ref{eqrun}) for $c=1$ as shown
in the Figure~\ref{ball}(b) where time between every consecutive transition is
$2\sqrt{2\ell/g}$.
The infinite run in this example seems natural since we assume the restitution
coefficient $c=1$, and under this unrealistic situation we expect the ball to
bounce indefinitely.  
However, given the generality of the model of hybrid automata the time
divergence of a run is not always guaranteed.
As an example consider again the bouncing ball system now with restitution
coefficient $0 < c < 1$.
In this case the time of the run~(\ref{eqrun}) is $\tot{r} = t_1(1+c)/(1-c)$ is
finite for any $ 0 < c < 1$.
Runs that are not time-diverging, on an intuitive level, are not physically
realizable since they execute infinitely many discrete transitions in a finite
amount of time. 
Assuming the possibility of realizing infinitely many discrete actions in a
finite time often lead to paradoxical situations, commonly known as Zeno's
paradoxes, and the runs that do not diverge also go by the name of Zeno runs. 
We call a hybrid automaton \emph{non-Zeno} if it does not permit any Zeno run. 
We will later see that the ability of hybrid automata to model Zeno runs
often cause difficulty in their analysis.

\subsection{Composition of a Network of Hybrid Automata}
While modeling a complex hybrid system using a hybrid automata, it is often 
convenient to represent various components of the system as a network of
hybrid automata $\Cc = \set{\Hh^1, \Hh^2, \ldots, \Hh^n}$ that communicate with
each other using shared variables and action.
Specifying a system as a composition of various subsystems offer two main
advantages, namely abstraction and modularity. 
The first advantage (abstraction) is that it allows the system designer to
concentrate on the details of one subsystem at a time without getting
overwhelmed by the complexity of the interaction of this subsystem with other. 
The second advantage (modularity) is that in a system designed in this fashion,
it is easy to add, remove, and modify subsystems.
The semantics of such a network can also be given as a single hybrid automaton
$\Hh$, called the product automaton of $\Cc$, whose states are products of
states of individual component automata.  
We define this construction as the following.
\begin{definition}[Composition]
  Let $\Cc = \set{\Hh^1, \Hh^2, \ldots, \Hh^n}$ be a network of hybrid automata
  where for each $1 \leq i \leq n$ let $\Hh^i$ be 
  $(M^i, M_0^i, \Sigma^i, X^i, \Delta^i, I^i, F^i, V_0^i)$.
  For an action $a \in \cup_{i=1}^n \Sigma_i$ we define 
  $E(a) \rmdef \set{i \::\: a \in \Sigma^i}$.
  The product automata $\Hh_1 \otimes \Hh_2 \otimes \cdots
  \Hh_n$ of $\Cc$ is defined as a hybrid automaton $H = (M, M_0, \Sigma, X,
  \Delta, I, F, V_0)$ where 
  \begin{itemize}
  \item $M = M^1 \times M^2 \times \cdots M^n$,
  \item $M_0 = M^1_0 \times M^2_0 \times \cdots M^n_0$,
  \item $\Sigma = \Sigma^1 \cup \Sigma^2 \cup \ldots \Sigma^n$,
  \item $X = X^1 \cup X^2 \cup \ldots X^n$,
  \item $\Delta {\subseteq} (M {\times}\pred(X) {\times} \Sigma {\times}
    \pred(X{\cup} X') {\times} M)$  is defined s.t. $((m_1, \ldots, m_n), g, a, j,
    (m_1', \ldots, m_n')) \in \Delta$ if and only if  for all $i \not \in E(a)$ we have that
    $m_i =  m_i'$ and  for all $i \in E(a)$ there exists a
    transition $(m_i, g_i, a, j_i, m_i')$ such that $g = \wedge_{i \in E(a)}
    g_i$ and $j = \wedge_{i \in E(a)} j_i$.
  \item $I$ is such that $I(m_1, \ldots, m_n) = \wedge_{i=1}^{n} I^i(m_i)$; 
  \item $F$ is such that $F(m_1, \ldots, m_n)(x) = F^i(m_i)(x)$ if $x {\in} X^i$;
    and
  \item $V_0$ is such that $V_0 = \wedge_{i=1}^{n} V_0^i$.
  \end{itemize}
\end{definition}
As an example of modeling a system using a composition of a network of hybrid
automata, we consider the job-shop scheduling problem modeled as a collection of
hybrid automata.  
In the next section, we show that solving the job-shop problem reduces to
solving a verification problem (reachability) over the resulting hybrid
automata. 
\begin{example}[Job-shop Scheduling Problem]
  The job-shop scheduling problem is an important optimization problem studied
  frequently in both computer science as well as in operations research.
  It consists of a finite set $\jJ=\set{j_1, \dots, j_n}$ of jobs to be
  processed on a finite set $\mM=\set{m_1, \dots, m_k}$ of machines. 
  There is a strict precedence requirement between the jobs given as a strict
  partial order $\prec$ over the set of jobs in $\jJ$. 
  A mapping $\zeta : \jJ \to 2^\mM$ specifies the set of machines where
  a job can be executed, while the function $\delta: \jJ \to \Rplus$ specify the
  time duration of a job. 
  We can model the job-shop scheduling problem using a network of hybrid
  automata where each job and each machine is specified using a separate hybrid
  automaton.  
  We have the following constraints on the job execution:
  i) a job $j$ can be executed iff all jobs in its precedence, 
  $j{\downarrow}=\{j' \::\: j' \prec j\}$, 
  have terminated;
  2) each machine $m \in \mM$ can process atmost one job at a time; and
  3) a job, once started, cannot be preempted.

  \noindent{\bf Modeling Jobs}. We model each job $j_i \in \jJ$ as a hybrid automaton $H_i$ 
  with three modes $U_i$ (unscheduled), $S_i$ (scheduled), and $F_i$ (finished)
  where $U_i$ being the initial mode.
  With each automaton $H_i$ we associate two variables: variable $x_i$, measuring
  the time while the job $j_i$ is being executed on a machine; and variable
  $\done_i$ with values $0$ and $1$ denoting whether the job is unfinished ($0$) or
  finished ($1$). 
  For each job $j_i$ the initial valuation of variable $x_i$ is $0$, while the
  valuation for $\done_i{=}0$. 
  For each mode $m \in \set{U_i, S_i, F_i}$ we have that $F(m)(\done_i) = 0$ and
  $F(S_i)(x_i) = 1$ (to measure time spent during processing of the job)  and
  $F(U_i)(x_i) = 0$ and $F(F_i)(x_i) = 0$.  
  The transition from a mode $U_i$ to $S_i$ with action $\beg_i$ is guarded
  by the condition that all of the preceding jobs according to $\prec$ has
  been finished, i.e. $\wedge_{k: k \prec i} (\done_k =1)$.
  The transition from a mode $S_i$ to $F_i$ with action $\fin_i$ is guarded
  by predicate $\done_i' = \delta(j_i)$ specifying that job $j_i$ takes exactly
  $\delta(j_i)$ time units , and the jump of this transition includes 
  $\done_j' = 1$.
  
  \noindent{\bf Modeling Machines}. 
  We model each machine $m_i \in \mM$ using a hybrid automaton with no variable
  and $k+1$ modes where $k$ is the number of jobs that can be scheduled to this
  machine: there is a unique mode $I_i$ (idle), and for each job $j_j$ that can
  be scheduled to this machine, i.e. $m_i \in \zeta(j_i)$ there is a mode
  $P_{i,j}$ (corresponding to processing job $j_j \in \jJ$ on machine $m_i \in
  \mM$).  
  For each mode $P_{i,j}$ there is a transition from $I_i$ to $P_{i, j}$ with
  action $\beg_j$ and a transition from $P_{i,j}$ to $I_i$ with action $\fin_j$
  denoting the scheduling and the finishing, respectively, of job $j_j$ on machine $m_i$. 
  Since there are no variables associated with these automata the guard and the
  jump predicate of these transitions is simply $\top$.
  
  As an example of such modeling, consider the job-shop problem with $J=\{j_1,j_2\}$,  
  $M=\{m_1\}$, $\zeta(j_1)=\zeta(j_2)=m_1$,  $j_1 \prec j_2$, and $\delta(j_1) =
  3$ and $\delta(j_2) = 4$. 
  Figure~\ref{job2} shows hybrid automata $\Hh_{j_1}, \Hh_{j_2}$, and $\Hh_{m_1}$
  corresponding to the jobs $j_1$ and $j_2$, and the machine $m_1$ respectively.
  This figure also shows the composition of these automata $\Hh_{j_1} \otimes
  \Hh_{j_2} \otimes \Hh_{m_1}$  representing the hybrid automata corresponding
  to the complete job-shop problem. 
\end{example}

\section{Formal Verification of Hybrid Systems}
\label{sec:verification}
Formal modeling and verification of systems is the set of techniques that employ
rigorous mathematical reasoning to analyze properties of a system. 
In this article we concentrate on model checking---a formal verification
framework introduced by Clarke, Sifakis and Emerson~\cite{CES09}---that, given a
formal description of a system and its specification, systematically
verifies whether the specification holds for the system model.  
Since, by definition the states of a dynamical system changes with time,
classical propositional logic is not sufficient to reason with temporal
properties of such dynamical systems.
Temporal logics extend propositional or predicate logics by modalities 
that are useful to capture the change of behaviour of a system over time. 
Manna and Pnueli~\cite{MP92,TM0} were the first one to propose and promote the
use of temporal logic to specify properties of dynamical systems in the context
of system verification.
\emph{Linear temporal logic} (LTL)~\cite{MP92}, \emph{computation tree logic}
(CTL) and its generalization CTL$^*$~\cite{MP92,CGP99}, and modal
$\mu$-calculus~\cite{Eme96} are some of the popular temporal logics used for the
system specification. 
Timed and weighted extensions of these logics  e.g. metric temporal logics
(MTL and MITL~\cite{OW08}), duration calculus (DC)~\cite{GNRR93}, and 
weighted logics~\cite{BBR04,Bou06} have also been proposed to specify more
involved quantitative properties of hybrid dynamical systems.

In this article we limit the discussion to simple qualitative properties of
hybrid systems that broadly can be classified into the following two broad
categories~\cite{MP87}: 
\begin{itemize}
\item 
  The \emph{reachability or guarantee properties}, that ask whether the system
  can reach a  configuration satisfying certain property $p$? 
  (symbolically, we write $\Diamond p$ and we say \emph{eventually} $p$); and
\item 
  The \emph{safety properties} that ask whether the system can stay forever in
  configurations satisfying certain property $p$? 
  (symbolically, we write $\Box p$ and we say \emph{always} or \emph{globally} $p$).
\end{itemize}
The linear temporal logic, LTL, provides a formal language to specify more
involved nesting of such properties with ease.  
We begin this section (Section~\ref{sec:LTL}) by introducing Kripke structures
that provide a way to mark states of the hybrid automata with properties of
interest, and present the syntax and semantics of LTL that are interpreted over
Kripke structures. 
In Section~\ref{sec:undec} we formally introduce LTL model-checking problem for
hybrid automata, and show that in general this problem is undecidable.
On a positive note, in Section~\ref{sec:finite}, we show that LTL model-checking
can be algorithmically solved for finite Kripke structures.
Finally, in Section~\ref{sec:finiteBisum} we introduce the notion of
bisimulation, and show that the existence of a finite bisimulation implies the
decidability of LTL model-checking problem.

\subsection{Hybrid Kripke Structures and Linear Temporal Logic}
\label{sec:LTL}
The formal specification of the underlying system begins by identifying key
properties of interests (called atomic propositions) regarding the states of the
system under verification.  
Kripke structures provide a way to label the states of state-transition graphs
with such atomic propositions, and the linear temporal logic specifies
properties of the sequence of the truth values of these propositions, called 
\emph{traces}, for the runs of corresponding transition system.
Hence, before we introduce linear temporal logic LTL we need to introduce Kripke
structures and their corresponding hybrid extension, and the concept of traces. 

\begin{definition}[Hybrid Kripke Structure]
  A \emph{Kripke Structure} is a tuple $(\Tt, P, L)$ where:
  \begin{itemize}
  \item 
    $\Tt = (S, S_0, \Sigma, \Delta)$ is a state transition graph, 
  \item 
    $P$ is a finite set of \emph{atomic propositions}, and 
  \item 
    $L: S \rightarrow 2^{P}$ is a labeling function that labels every state
    with a subset of $P$. 
  \end{itemize}
  Similarly, we define a \emph{Hybrid Kripke Structure} as a tuple $(\Hh, P, L)$
  where:  
  \begin{itemize}
  \item 
    $\Hh = (M, M_0, \Sigma, X, \Delta, I, F, V_0)$ is a hybrid automaton,
  \item 
    $P$ is a finite set of \emph{atomic propositions}, and 
  \item 
    $L: M \rightarrow 2^{P}$ is a labeling function that labels every mode
    with a subset of $P$. 
  \end{itemize}
  Observe that the semantics of a hybrid Kripke structure is a Kripke
  structure. 
\end{definition}
Let us fix a hybrid Kripke structure $(\Hh, P, L)$ and its semantics Kripke
structure $(\sem{\Hh}, P, L)$ for the rest of this section.
When the set of propositions and labeling function is clear from the context, we
use the terms state transition graph and Kripke structure, and the terms
hybrid Kripke structure and hybrid automaton interchangeably. 

%The semantics of a hybrid Kripke structure is given by an infinite state transition system.  
%Given an action $a \in \Sigma$ and a state $s$ we write $\post(s,a)$ for the set
%of states that are  reachable from $s$ on $a$ i.e. $\post(s, a) = \set{s' \mid
%  (s,a,s') \in \Delta}$.  $\post(s)=\bigcup_{a \in \Sigma} \post(s,a)$. 
%We write $\RUNS(T)$ for the set of all runs of $T$ starting from some initial state.
%Denote by $\RUNS(T,s_j)$ the runs starting from a state $s_j$. 
%For a run $r \in \RUNS(T)$ we define $\LAST(r)$ as the last state of the run
%$r$. 
%Given a state $s_j$,  we define 
%$Lan(s_j)=\{a_{j+1}a_{j+2}\dots a_n \mid s_j, a_{j+1}, s_{j+1}, a_{j+2}, s_{j+2}, \ldots, a_n, s_n \in \RUNS(T,s_j)\}$.
%The set of all traces is given by $Lan(\Kk)=\bigcup_{s \in M_0}Lan(s)$.  

Given a hybrid Kripke structure $(\Hh, P, L)$ and an infinite run $r=\seq{(m_0,
  \nu_0), (t_1, a_1), (m_1, \nu_1), \ldots, (m_n, \nu_n), \ldots}$ of $\Hh$, we
define a trace corresponding to $r$, denoted as 
$\Trace(r)$, as the sequence $\seq{L(m_0), L(m_1), L(m_2),  \dots L(m_n), \ldots}$.
Let $\Trace(\Hh, P, L)$ be the set of traces of the Hybrid Kripke Structure $\Hh$.
For a trace $\sigma = \seq{P_0, P_1, \ldots, P_n, \ldots}  \in \Trace(\Hh, P, L)$ we write
$\sigma[i]=\seq{P_i, P_{i+1}, \ldots}$ for the suffix of the trace starting at
the index $i \geq 0$. 

Now we are in position to define the syntax and semantics of linear temporal logic.
\begin{definition}[Linear Temporal Logic (Syntax)]
  The set of valid LTL formulas over a set $P$ of atomic propositions can be
  inductively defined as the following:
\begin{itemize}
\item $\top$ and $\bot$ are valid LTL formulas;
\item if $p \in P$ then $p$ is a valid LTL formula;
\item if $\phi$ and $\psi$ are valid LTL formulas then so are $\neg \phi$, $\phi
  \wedge \psi$ and $\phi \lor \psi$; 
\item if $\phi$ and $\psi$ are valid LTL formulas then so are $\bigcirc\phi$,
  $\Diamond \phi$, $\square \phi$, and $\phi \until \psi$.
\end{itemize}
\end{definition}
We often use $\phi \Rightarrow \psi$ as a shorthand for $\neg \phi \lor \psi$.
Before we define the semantics of LTL formula formally, let us give an informal
description of the temporal operators $\bigcirc$, $\Diamond$, $\square$, and
$\until$. 
LTL formulas are interpreted over traces of (Hybrid) Kripke structures. 
The formula $\bigcirc \phi$, read as next $\phi$, holds for a trace
$\sigma=\seq{P_0, P_1, P_2, \ldots}$ if $\psi$ holds for the trace $\sigma[1]$. 
The formula $\Diamond \phi$, read as eventually $\phi$, holds for a trace
$\sigma=\seq{P_0, P_1, P_2, \ldots}$ if there exists $i \geq 0$ such that the
formula $\psi$ holds for the trace $\sigma[i]$. 
The formula $\square \phi$, read as globally or always $\phi$, holds for a trace
$\sigma=\seq{P_0, P_1, P_2, \ldots}$ if for all $i \geq 0$ the formula $\psi$
holds for traces $\sigma[i]$. 
Finally, the formula $\phi \until \psi$, read as $\phi$ until $\psi$, holds for
a trace $\sigma=\seq{P_0, P_1, P_2, \ldots}$ if there is an index $i$ such that
$\psi$ holds for the trace $\sigma[i]$, and for every index $j$ before $i$ the
formula $\phi$ holds for the trace $\sigma[j]$, i.e the formula $\phi$ holds until
formula $\psi$ holds.
\begin{definition}[Linear Temporal Logic (Semantics)]
  For a trace $\sigma=\seq{P_0, P_1, P_2, \ldots}$ of a (Hybrid) Kripke
  structure we write $\sigma \models \phi$ to say that the trace
  $\sigma$ satisfies the formula $\phi$. 
  The satisfaction of LTL formulas is defined as follows:
  \begin{itemize}
  \item 
    $\sigma \models \top$ and $\sigma \not \models \bot$;
  \item 
    $\sigma \models p$ if $p \in P_0$;
  \item
    $\sigma \models \neg \phi$ if $\sigma \not \models \phi$;
  \item
    $\sigma \models \phi \wedge \psi$ if $\sigma \models \phi$ and $\sigma
    \models \psi$;
  \item
    $\sigma \models \phi \lor \psi$ if $\sigma \models \phi$ or $\sigma
    \models \psi$;    
  \item
    $\sigma \models \bigcirc\phi$ if $\sigma[1] \models \phi$;
  \item
    $\sigma \models \Diamond \phi$ if there exists $i \geq 0$ such that $\sigma[i]
    \models \phi$; 
  \item
    $\sigma \models \square \phi$ if for all $i \geq 0$ we have that $\sigma[i]
    \models \phi$; and
  \item
    $\sigma \models \phi \until \psi$ if there exists $i \geq 0$ such that
    $\sigma[i] \models \psi$, and for all $0 \leq j < i$  or
    $\sigma[j]\models \phi$. 
  \end{itemize}
  For a (hybrid) Kripke structure $(\Hh, P, L)$, and an LTL formula $\phi$ we
  say that $(\Hh, P, L) \models \phi$ if for all $\sigma \in \Trace(\Hh, P, L)$
  we have that $\sigma \models \phi$.
\end{definition}

Lamport~\cite{Lam77} observed that most of the system specifications can be
classified in safety properties (\emph{something will not happen}) and liveness
properties (\emph{something must happen}). 
Manna and Pnueli~\cite{MP87} further refined the class of specifications
starting from reachability and safety properties to introduce a hierarchy of
temporal properties using nesting of LTL operators, for instance
\begin{itemize}
\item 
  The \emph{recurrence properties} that ask whether the system can infinitely
  often visit configurations satisfying certain property $p$? 
  (symbolically, we write $\Box \Diamond p$ and we say \emph{infinitely often}
  $p$); and 
\item 
  The \emph{persistence properties} that ask whether the system visits
  configurations not satisfying a certain property $p$ only finitely often?  
  (symbolically, we write $\Diamond \Box p$ and we say \emph{eventually always}
  $p$).
\end{itemize}
Some examples for expressing reachability, safety, and liveness properties using
LTL are shown in the following example.
\begin{example}
  \label{lift}
  As an example let us write LTL specifications for an elevator serving $k$
  different floors.
  Let $op_i, fl_i$ and $req_i$ be atomic propositions representing the
  situations that ``the door at floor $i$ is open'', `` the lift is at floor $i$
  and is not moving'' and  ``there is a request for the lift to move to the
  $i$th floor'' respectively.  
  The following are some specifications in English and their LTL counterparts: 
  \begin{enumerate}
  \item Reachability property : 
    {\it The lift will visit the ground floor sometime}. 
    \[ 
    \phi_1 \rmdef \Diamond fl_0.
    \]
  \item Safety property : 
    {\it The door of the lift is never open at a floor if the lift is not
      present there}. 
    \[ 
    \phi_2 \rmdef \Box\left( \bigwedge_{i=0}^k(\neg fl_i \Rightarrow \neg op_i)\right).
    \]  
  \item Recurrence property : {\it The lift keeps coming back to the ground floor}.
    \[ 
    \phi_3 \rmdef \Box \left(\neg fl_0 \Rightarrow \Diamond fl_0 \right) \wedge \Box \Diamond
    fl_0.
    \] 
  \item Persistence property : 
    {\it Eventually always a requested floor will be eventually
      served}. 
    \[ 
    \phi_4 \rmdef\Diamond\Box \left(\bigwedge_{i=0}^k (req_i \Rightarrow \Diamond 
    fl_i)\right).
    \] 
\end{enumerate}
\end{example}
For a detailed overview of LTL for system specification, we refer the reader
to~\cite{MP87,MP92,CGP99,BK08}. 
\subsection{LTL Model Checking  for Hybrid Automata}
\label{sec:undec}
\begin{figure*}[t]
\begin{center}
\begin{tikzpicture}[->,>=stealth',shorten >=1pt,auto,node distance=1.8cm,
  semithick]
  \tikzstyle{every state}=[fill=purple!30!white,minimum size=3em,rounded rectangle]
  \node[fill=white] at (-6,0) (A1){};
  \node[initial,initial text={},state,fill=gray!30] at (-4,0) (A) {$\begin{array}{c}l_i \end{array}$} ;
  \node[state,fill=gray!20] at (-1,0) (B) {$\begin{array}{c} A_i\end{array}$} ;
  \node[state,fill=gray!20] at (3,0) (C) {$\begin{array}{c} B_i \end{array}$} ;
  \node[state,fill=gray!20] at (8,0) (D) {$\begin{array}{c}l_j\end{array}$} ;
  \path(A1) edge  node [above] {$y=0?$} (A);
  \path(A) edge  node [above] {$x_1=1?$} (B);
  \path(A) edge  node [below] {$z'=0$} (B);
  \path(B) edge node [above] {$x_1>1? x_2,y<1?$} (C);
  \path(B) edge node [below] {$x'_1,z'_1=0$} (C);
  \path(C) edge node [above] {$z=z_1?,y=1?,x_2<1?$} (D);
  \path(C) edge node [below] {$y'=0$} (D);
  \path(A) edge [loop above] node [above] {$x_2=1?,x_2'=0$} (A);
  \path(B) edge [loop above] node [above] {$x_2=1?,x_2'=0$} (B);
     \path(C) edge [loop above] node [above] {$x_2=1?,x_2'=0$} (C);
   \end{tikzpicture}
\caption{Module simulating $l_i$: increment $c$, goto $l_j$}
\label{inc}
\end{center}
\end{figure*}
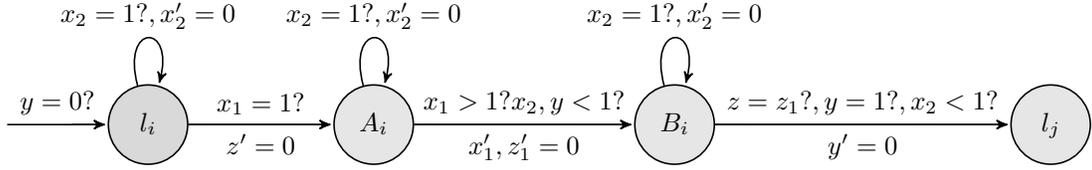
LTL model-checking problem for hybrid automata can be formally stated in the
following manner. 
\begin{definition}[LTL Model-Checking]
  Given a system modeled as a (Hybrid) Kripke structure $(\Hh, P, L)$, and a
  specification written as an LTL formula $\phi$, the {\it LTL model-checking}
  problem is to decide whether all traces of $\Hh$ satisfy $\phi$, i.e.  $(\Hh,
  P, L) \models \phi$.
  Moreover, if the system does not satisfy the property give a counterexample
  (run of the system) violating the property.
\end{definition}

\begin{example}
  Consider the following Kripke structure $\Tt$ with set of atomic propositions
  $\set{p, q}$.
  We are depicting the labeling function by writing the set of propositions
  inside the states., and we omit other non-relevant details.
  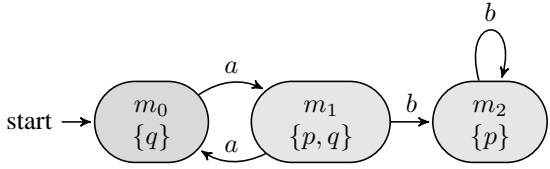
\begin{figure}[t]
\begin{center}
  
  \begin{tikzpicture}[->,>=stealth',shorten >=1pt,auto,node distance=2.8cm,
  semithick,scale=0.9]
  \tikzstyle{every state}=[fill=gray!20!white,minimum size=3em,rounded rectangle]
  \node[initial,state,fill=gray!30] at (-2.5,0) (A) {$
      \begin{array}{c}m_0\\ \{q\}\end{array}$} ;
  \node[state,fill=gray!20] at (0,0) (B) {$\begin{array}{c}m_1\\ \{p, q\}\end{array}$} ;
    \node[initial where=right,state,fill=gray!20] at (2.5,0) (C)
    {$\begin{array}{c}m_2\\ \{p\}\end{array}$} ;
\path(A) edge [bend left=30] node [above] {$a$} (B);
  \path(B) edge [bend left=30] node [above] {$a$} (A);
  \path(B) edge node {$b$} (C);
  \path (C) edge[loop above] node {$b$} (C);
  \end{tikzpicture}
\caption{A Kripke structure $\TS$.} 
\label{ts1}
\end{center}
\end{figure}
  Let us consider the LTL formulas $\phi_1= \Diamond (p \wedge \neg q)$ and $\phi_2 =
  \Box q \lor \Diamond \Box p$.
  Observe that $\Tt \not \models \phi_1$ as is clear from the counterexample
  $r = \seq{m_0, a, m_1, a, m_0, \ldots}$ as it never visits the configuration 
  satisfying $(p \wedge \neg q)$ as is clear from its trace 
  $\Trace(r)=\{q\}\{p,q\}\set{q}\{p,q\}$. 
  On the other hand, it is easy to verify that $\Tt$ satisfies $\phi_2$ as any
  run of $\Tt$ either never visits $m_2$ (and in that case satisfies $\Box q$,
  or it eventually visits $m_2$ and never leaves it (and thus satisfies $\Diamond
  \Box p$). 
\end{example}
\begin{example}[Job-Shop Scheduling Revisited]
  Consider the job-shop scheduling problem modeled as a network of hybrid
  automata in Figure~\ref{job2}.
  Consider the atomic propositions $j_1.\fin$ and $j_2.\fin$ that are true
  only in modes $F_1$ and $F_2$.
  The counterexample produced in model-checking LTL property 
  $\neg (\Diamond (j_1.\fin \wedge j_2.\fin))$ gives a valid schedule for the
  job-shop scheduling problem. 
\end{example}

Next, we show that LTL model-checking problem for hybrid Kripke
structures is undecidable. 
To prove this result, we show a reduction from a well-known undecidable
problem of reachability (halting) for two-counter \emph{Minsky
  machines}~\cite{Min67}.  

A Minsky machine $\Aa$ is a tuple $(L, C)$ where:
${L = \set{\ell_0, \ell_1, \ldots, \ell_n}}$ is the set of
instructions. There is a distinguished terminal instruction  $\ell_n$ called
HALT. 
${C = \set{c_1, c_2}}$ is the set of two \emph{counters};
the instructions $L$ are one of the following types:
\begin{enumerate}
\item (increment $c$) $\ell_i : c := c+1$;  goto  $\ell_k$,
\item (test-and-decrement $c$) $\ell_i$ : if $(c >0)$ then
  $(c:= c - 1;$ goto $\ell_k)$ else goto $\ell_m$,
  \item (Halt) $\ell_n:$ HALT.
\end{enumerate}
where $c \in C$, $\ell_i, \ell_k, \ell_m \in L$.

A configuration of a Minsky machine is a tuple $(\ell, c, d)$ where
$\ell \in L$ is an instruction, and $c, d$ are natural numbers that specify the value
of counters $c_1$ and $c_2$, respectively.
The initial configuration is $(\ell_0, 0, 0)$.
A run of a Minsky machine is a (finite or infinite) sequence of
configurations $\seq{k_0, k_1, \ldots}$ where $k_0$ is the initial
configuration, and the relation between subsequent configurations is
governed by transitions between respective instructions.
The run is a finite sequence if and only if the last configuration is
the terminal instruction $\ell_n$.
Note that a Minsky machine has exactly one run starting from the initial
configuration. 
The \emph{halting problem} for a Minsky machine asks whether 
its unique run ends at the terminal instruction $\ell_n$.
It is well known~(\cite{Min67}) that the halting problem for
two-counter Minsky machines is undecidable.
\begin{theorem}
\label{undec}
The LTL model-checking problem for hybrid Kripke structures is undecidable. 
\end{theorem}
\begin{proof}
Given a two counter machine ${\cal A}$, we construct a hybrid Kripke structure
${\cal H}$ and an LTL formula $\phi$ such that ${\cal H}\models \phi$ iff
${\cal A}$ halts. 
The modes of ${\cal H}$ are labeled with the labels $l_i$ of instructions. There
is a unique mode of ${\cal H}$ labeled with atomic proposition ``HALT'' which
corresponds to the terminal instruction of $\Aa$. 
The increment, decrement and test instructions are encoded by suitable modules
in ${\cal H}$.  
The variables of ${\cal H}$ are $X=\{x_1,x_2,y,z,z_1\}$ with $F(m)$ for all
modes is defined as the following: 
\[
\dot{x_1}=1 \wedge \dot{x_2} = 1 \wedge \dot{y} = 1 \wedge \dot{z}=1 \wedge
\dot{z_1}=2.
\] 
The initial mode is labeled by $l_0$, the label of the first instruction. 
The values of the counters $c,d$ are encoded as $x_1=\frac{1}{2^c}$ and $x_2=\frac{1}{2^d}$.
After the execution of each instruction, $x_1, x_2$ will contain the current
values of counters $c,d$  encoded in the above form. 
For instance, if we have $x_1=\frac{1}{2^c}$, $x_2=\frac{1}{2^d}$ before 
incrementing counter $c$, then at the end of simulating the increment
instruction, we will have $x_1=\frac{1}{2^{c+1}}$ and $x_2=\frac{1}{2^d}$.

We illustrate here the case of the increment instruction $l_i$ : increment $c$
and goto $l_j$.  
The case for the decrement instruction is similar, and hence omitted.
Mode $l_i$ is entered with $y=0$, $x_1=\frac{1}{2^c}$ and
$x_2=\frac{1}{2^d}$. On entering mode $A_i$, we have $x_1=1, y= 1-\frac{1}{2^c},
x_2=\frac{1}{2^d}+(1-\frac{1}{2^c})$  or $x_2=1-\frac{1}{2^c}-\eta$ if
$\frac{1}{2^d}+\eta=1$, $\eta \leq 1-\frac{1}{2^c}$ and $z=0$. 
Mode $B_i$ can be entered if $x_2,y<1$ and $x_1 > 1$. 
Assume $k>0$ units of time was spent at mode $A_i$.
This gives $y=1-\frac{1}{2^c}+k, x_2=\frac{1}{2^d}+(1-\frac{1}{2^c})+k$ (or
$1-\frac{1}{2^c}-\eta+k$, or $1-\eta'$ if $1-\frac{1}{2^c}-\eta+\eta'=1$, $\eta'
\leq k$), $z=k$, $x_1=0,z_1=0$ on entering mode $B_i$. 
We can reach mode $l_j$ only if the values of $z$ and $z_1$ are the same.
Assume $l$ units of time was spent at $B_i$. 
Then $z=k+l, z_1=2l$,
$x_2=\frac{1}{2^d}+(1-\frac{1}{2^c})+k+l,x_1=l,y=1-\frac{1}{2^c}+k+l$. 
To satisfy the constraints $z=z_1,y=1$,  we have $k=l$ and
$k+l=2k=\frac{1}{2^c}$ giving $x_1=\frac{1}{2^{c+1}}, x_2=\frac{1}{2^d}$,$y=0$
at $l_j$.  

The LTL formula $\phi=l_0 \wedge \Diamond$ HALT will be satisfied by 
${\cal H}$ iff $\Aa$ halts. This shows that LTL model checking of 
hybrid Kripke structures is undecidable.
\end{proof}

\subsection{LTL Model-Checking for Finite Kripke Structures}
\label{sec:finite}
As we discussed in previous section the LTL model-checking problem is
undecidable for general hybrid automata.
However, for finite Kripke structures Wolper, Vardi, and Sistla~\cite{WVS83}
developed an elegant automata-theoretic algorithm for solving the LTL
model-checking problem.  
The algorithm exploits the connection between LTL formulas and a type of
$\omega$-automata---automata that extend the theory of finite automata to
infinite inputs---called B\"uchi automata~\cite{Bu62,GTW02}.   
The syntax for the B\"uchi automata specifies a finite state transition graph
$\Tt$ along with a set $F$ of accepting states, and the semantics of B\"uchi
automata restricts the set of valid runs to the runs of $\Tt$ that visit $F$
infinitely often. 
In general B\"uchi automata are closed under all Boolean operations including
union, intersection, and complementation, however deterministic variant of
B\"uchi automata is not closed under complementation. 
Emptiness checking for B\"uchi automata can be decided efficiently (linear in
time) by analyzing strongly connected components of $\Tt$. 

The LTL model-checking problem exploits the following connection between linear
temporal logic  and B\"uchi automata.
\begin{theorem}[LTL-to-B\"uchi Automata~\cite{WVS83}]
  \label{thm:LTL2Buchi}
  For every LTL formula $\phi$ we can effectively construct a
  finite (B\"uchi) automaton $\Aa_\phi$ (of size exponential in $\phi$) such
  that words recognized by $\Aa_\phi$ are precisely the set of traces that
  satisfy $\phi$.   
\end{theorem}

Based on this result, the LTL model checking for a finite Kripke structure $\Kk$
can be performed in the following manner: 
\begin{enumerate}
\item 
  Construct a B\"uchi automaton $\Aa_{\neg \phi}$ corresponding to the negation of
  the LTL property. 
\item
  Construct the composition $\Kk {\otimes} \Aa_{\neg \phi}$ of the Kripke
  structure $\Kk$ with the B\"uchi automaton $\Aa_{\neg \phi}$.
\item 
  If the B\"uchi automaton $\Hh {\otimes} \Aa_{\neg\phi}$  is empty, then return
  ``TRUE''  
\item
  Else, return a lasso-shaped (a finite prefix followed by a cycle that contains
  an accepting state) infinite run accepted by $\Hh {\otimes} \Aa_{\neg\phi}$ as a
  counter-example.
\end{enumerate}
The correctness of this algorithm follows from the observation that the set of
traces for this composition $\Kk {\otimes} \Aa_{\neg \phi}$ characterize the set
of traces that are generated by $\Kk$ that do not satisfy $\phi$.   
Hence, the Kripke structure $\Kk$ satisfies the LTL property $\phi$ if and only
if  $\Hh {\otimes} \Aa_{\neg\phi}$  is empty. 

\begin{theorem}[LTL model-Checking for Finite Structures~\cite{SC85}]
  \label{keythm}
  LTL model checking problem for finite Kripke structures is decidable in
  PSPACE. 
\end{theorem}

LTL model-checking for finite Kripke structures is implemented by a number of
mature tools, notably SPIN~\cite{TSPIN} and NuSMV~\cite{TNuSMV}, and has been
applied to a number of practical case-studies~\cite{TSPIN,TNuSMV}.

%Let us revisit example \ref{lift} where we wrote some specifications 
%for an elevator system.  The figure on top in Figure \ref{lift1} gives a kripke structure ${\mathcal H}$ modelling 
%the elevator. For each property $\phi_i$ given in example \ref{lift}, checking whether ${\mathcal H} \models \phi_i$ 
%can be done using the steps described above. 

%\input{lift1}
%Note that ${\mathcal H'}$ given in figure \ref{lift1} is infact a finite state transition graph. 
%The bottom figure  in Figure \ref{lift1} gives the automaton ${\mathcal A}_{\neg \phi_3}$ 
%corresponding to $\neg \phi_3= \Diamond \Box \neg fl_0 \vee \Diamond(\neg fl_0 \wedge \Box \neg fl_0)$. 
%The atomic propositions corresponding to ${\mathcal A}_{\neg \phi_3}$  is same as those for 
%$\Hh$. $X_0$ is any subset of the atomic propositions containing $fl_0$, while 
%$X_1$ is any subset of the atomic propositions that does not contain $fl_0$. 
%Clearly, the composition ${\mathcal H} \times {\mathcal A}_{\neg \phi_3}$ is not empty;
%the trace $\{fl_0\}\{fl_0,req_1\}$ followed by $\{fl_1,op_1\}\{fl_1\}$ repeated 
%infinitely many times is allowed in the composition. 
%Therefore, $\Hh \nvDash \phi_3$. 

\subsection{Finite Bisimulation and Decidability} 
\label{sec:finiteBisum}
\label{bisim}
In this section we introduce the concept of bisimulation relation between two
Kripke structures, and show that for two bisimilar systems (systems having a
bisimulation relation between their states) we have that both systems 
have the same set of traces, and hence precisely the same set of LTL formulas
are satisfied by both of them.
Using this idea, we show that if for a given hybrid Kripke structure $\Hh$ there
exists a bisimulation relation with some finite state Kripke structure $\Kk$,
then the problem of LTL model-checking for $\Hh$ can be reduced to the decidable
problem of LTL model-checking for finite Kripke structure $\Kk$.

We say that a Kripke structure $\Kk' = (\Tt', P, L')$ can \emph{simulate} a Kripke
structure $\Kk= (\Tt, P, L)$ if every step of $\Kk$ can be matched (with 
respect to atomic propositions) by one or more steps of $\Kk'$.
A Bisimulation equivalence denotes the presence of a mutual simulation between
two structures $\Kk$ and $\Kk'$.
Formally, bisimulation relation in the following manner.
\begin{definition}[Bisimulation Relation]
  Let $\Kk = (\Tt = (S, S_0, \Sigma, \Delta), P, L)$ and $\Kk' = (\Tt = (S',
  S_0', \Sigma', \Delta'), P, L')$ be two Kripke structures.
  A bisimulation relation between $\Kk$ and $\Kk'$ is a binary relation $\Rr
  \subseteq S \times S'$ such that: 
  \begin{itemize}
  \item 
    every initial state of $\Tt$ is related to some initial state of $\Tt'$,
    and vice-versa, i.e. for every $s \in S_0$  there exists $s' \in S_0'$ such
    that $(s, s') \in \Rr$ and for every $s' \in S_0'$ there exists a $s \in
    S_0$ such that $(s, s') \in \Rr$;
  \item 
    for every $(s, s') \in {\mathcal R}$ the following holds:
    \begin{itemize}
    \item 
      $L(s)=L'(s')$,
    \item 
      every outgoing transition of $s$ is matched with some outgoing
      transition of $s'$, i.e. if $t \in \post(s)$ then there exists 
      $t' \in \post(s')$ with $(t, t') \in \Rr$, and  
    \item 
      every outgoing transition of $s'$ is  matched with some outgoing
      transition of $s$, i.e. if $t' \in \post(s')$ then there exists 
      $t \in \post(s)$ with $(t, t') \in \Rr$. 
    \end{itemize}
  \end{itemize}
\end{definition}
We say that $\Tt$ and $\Tt'$ (analogously, $\Kk$ and $\Kk'$) are bisimilar or
bisimulation equivalent, and we write $\Tt \sim \Tt'$, if there exists a
bisimulation relation $\Rr \subseteq S \times S'$. 

%\begin{example}
%  Figure \ref{bisim1} gives two transition systems which are not bisimilar. 
%  \input{bisim1.tex}
%  The transition systems in figure \ref{bisim1} model the flights of two
%  airlines. 
%  It is clear that for the first airline, it is possible to reach Shimla and
%  Delhi from Mumbai on  any day;  
%  the airline on the right however, has exactly one choice of destination from
%  Mumbai on any given day.  
%\end{example}

The following Proposition follows from the definition of bisimulation and the
semantics of LTL.
\begin{proposition}
  \label{trace-eq}
  If $\Tt \sim \Tt'$ then $\Trace(\Tt)=\Trace(\Tt')$.
  Moreover, if $\Tt \sim \Tt'$ then for every LTL formula $\phi$ we have that 
  $\Tt \models \phi$ if and only if $\Tt' \models \phi$.
\end{proposition}
\begin{proof}
  Let $\Tt \sim \Tt'$. 
  Using a simple inductive argument, one can show that for every run $a = \seq{s_0,
    a_1, s_1, a_2, \ldots}$ of $\Tt$ there is a run $r' = \seq{s_0', a_1', s_1', a_2',
    \ldots}$ of $\Tt'$ such that $L(s_i) = L'(s_i')$ for every $i \geq 0$.
  This implies that $\Trace(r) = \Trace(r')$ and hence $\Trace(\Tt) \subseteq
  \Trace(\Tt')$. 
  Similarly, we can show that $\Trace(\Tt') \subseteq \Trace(\Tt)$.
  Hence it follows that $\Tt \sim \Tt'$ implies  $\Trace(\Tt)=\Trace(\Tt')$.
  To prove the other part of the proposition, observe 
  % that two trace equivalent
  % structures $\Tt$ and $\Tt'$ always satisfy the same LTL formulae, since  
  %both have exactly the same set of traces, and 
  LTL formulae are interpreted over traces of structures, and since two
  bisimilar Kripke structures have the same set of traces, it follows that for
  every LTL formula $\phi$ we have that $\Tt \sim \Tt'$ implies that 
  $\Tt \models \phi$ if and only if $\Tt' \models \phi$.
\end{proof}
This proposition shows that LTL model checking problem can be reduced to solving
LTL model checking problem over a bisimilar Kripke structure.
We next show how to extend this idea to define bisimulation over the states of a
Kripke structure, and use it to produce a bisimilar Kripke structure with fewer
states. 
\begin{definition}[Bisimulation Relation on $\Kk$]
  Let $\Kk=(\Tt=(S, S_0, \Sigma, \Delta), P, L)$ be a Kripke structure. 
  A bisimulation on $\Kk$ is a binary relation $\Rr \subseteq S \times S$ such
  that for all $(s, s') \in \Rr$ we have that:
  \begin{itemize}
  \item 
    $L(s)=L(s')$;
  \item 
    if $t \in \post(s)$, then there exists an $t' \in \post(s')$ such that 
    $(t, t') \in \Rr$; 
 \item 
   if $t' \in \post(s')$, then there exists an $t \in \post(s)$ such that 
   $(t,t') \in \Rr$. 
\end{itemize}
\end{definition}
It is easy to see that a bisimulation relation $\Rr$ over the state space of
$\Kk$ is an equivalence relation. 
For a state $s \in S$ we write $[s]_\Rr$ for the equivalence class of $\Rr$
containing $s$.
We say that states $s, s' \in S$ are bisimulation equivalent, and we write 
$s \sim_{\Tt} s'$, if there exists a bisimulation relation $\Rr$ for $\Tt$ 
with $(s, s') \in \Rr$. 

Given a Kripke structure $\Tt$, we use a bisimulation relation $\Rr$ for
reducing the state space of $\Tt$ using the following quotient construction. 
\begin{definition}[Bisimulation Quotient]
  Given a Kripke structure $\Kk = (\Tt=(S, S_0, \Sigma, \Delta), P, L)$ and a
  bisimulation relation $\Rr \subseteq S \times S$ over $\Kk$, the bisimulation
  quotient $\Kk_\Rr$ is defined as a Kripke structure 
  $\Kk_\Rr=(\Tt_\Rr=(S_\Rr, S^0_{\Rr}, \Sigma_\Rr, \Delta_\Rr), P, L_\Rr)$
  where: 
  \begin{itemize}
  \item 
    the state space of $\Tt_\Rr$ is the quotient space of $\Tt$, i.e. 
    $S_\Rr = \set{[s]_\Rr \::\: s \in S}$;
  \item  
    the set of initial states is the set of $\Rr$-equivalence classes of the
    initial states, i.e. $S^0_\Rr=\set{[s]_\Rr \::\: s \in S_0}$;
  \item 
    $\Sigma_\Rr = \set{\tau}$;
  \item 
    Each transition $(s,a,s') \in \Delta$ induces a transition from $[s]_{\Rr}$
    to $[s']_{\Rr}$ in $\Delta_\Rr$, i.e. 
    $\Delta_\Rr=\set{([s]_\Rr, \tau, [s']_\Rr) \::\: (s, \alpha, s') \in
      \Delta}$, and
  \item 
    $L_\Rr$ is defined such that $L_\Rr([s]) = L(s)$~\footnote{Observe that the
      definition of bisimulation ensures that the state labeling $L_\Rr$ is well
      defined.}.
  \end{itemize}
  We say that a bisimulation quotient is \emph{finite} if there are finitely many
  equivalence classes of $\Rr$, i.e. $|S_\Rr| < \infty$.
\end{definition}
%  Note that for model checking purposes, we only need the traces 
%  of the transition system; the actions play no role. 
%  We now give an example to see the state space reductions possible by quotienting. 
%  Consider the situation at an airport with $n$ aircrafts at bay. 
%  Each aircraft is one of the two states : ready for take-off or 
%  not ready for take-off. Modeling this as a transition system $TS$, 
%  with $S=\{(a_1,a_2, \dots, a_n) \mid a_i \in \{0,1\}\}$, 
%  we have states corresponding to all $n$-tuples
%  of possibilities of aircraft $i$ being ready for take-off or not. 
%  However, intuitively, $L(s)$ denotes the number of aircrafts 
%  that are ready for take-off at any point of time. 
%  Figures \ref{q1}, \ref{q2} shows the system before and after quotienting for the case 
%  of $n=3$. 
%  
%  \input{q1}
%  \input{q2}
%
The proof of the following theorem is immediate from Proposition~\ref{trace-eq}
and Theorem~\ref{keythm}.
\begin{theorem}
  \label{thm:finiteKSDec}
  The existence of a finite bisimulation quotient for a hybrid Kripke structure
  imply the decidability of LTL model-checking problem.
\end{theorem}

\section{Decidable Subclasses of Hybrid Automata}
\label{sec:decidability}
Given the expressiveness of hybrid automata it is not surprising that simple
reachability questions are undecidable for general hybrid kripke structures.
In this section we discuss some prominent subclasses of hybrid automata for
which LTL model checking problem is decidable.
In the previous section we discussed that  showing the existence of a finite
bisimulation quotient guarantees decidable model-checking.
Timed automata were among the first hybrid automata shown to have decidable
model-checking using this approach. 
We begin this section by presenting timed automata and discuss this bisimulation
known as region-equivalence relation. 
We will also review multi-rate and rectangular hybrid automata
(Section~\ref{sec:rectangular}) that under certain restriction (initialized)
recover decidability of LTL model-checking  via reductions to similar problem on
timed automata.  
Finally, in Section~\ref{sec:MMS} we discuss a relatively simple class of hybrid
systems, called piecewise-constant derivative systems, that capture the essence
of undecidability and provide references to its variants that permit algorithmic
analysis.
\subsection{Timed Automata}
\label{sec:timedAutomata}
Timed automata, introduced by Alur and Dill~\cite{AD90,AD94}, is a popular
formalism to model real-time systems. 
A timed automaton is a hybrid automaton where all variables grow with a constant
and uniform rate (for all variables $x {\in} X$ we have that $\dot{x} = 1$) and
the only jump permitted during the discrete transitions is reset to zero. 
Moreover, the set of predicates permitted to appear as guard on transitions is
restricted to the following kind of octagonal predicates:  
\begin{equation}
g:= x \bowtie c | x-y \bowtie c |g \wedge g\label{guards}
\end{equation}
where $x,y$ are clock variables, $\bowtie \in \{<, \leq, =, >, \geq\}$ and  
$c \in \mathbb{N}$. 
We write $\zones(X)$ for this class of octagonal predicates over the set $X$.
Formally, we define a timed automata as a restriction of hybrid automata in the
following manner.
\begin{definition}[Timed Automata: Syntax]
  A timed automaton is a hybrid automaton 
  $\Tt= (M, M_0, \Sigma, X, \Delta, I, F, V_0)$ with the following restrictions:
  \begin{itemize}
  \item the transition relation 
    $\Delta \subseteq M \times \pred(X) \times \Sigma \times \pred(X\cup X')
    \times M$ is such that if $(m, g, a, j, m') \in \Delta$ then 
    \begin{itemize}
    \item the guard $g$ is of the form (\ref{guards}), i.e. $g  \in \zones(X)$
      and  
    \item the jump predicate $j$ only permits variable resets to zero, i.e. $j$
      is of the form
      \[ 
      \wedge_{x \in Y} (x'{=}0),
      \]
      for some $Y \subseteq X$. We denote such set $Y$ as $\rst(j)$.
    \end{itemize}
  \item 
    the mode-invariant function $I: M \to \pred(X)$ is such that for all
    $m \in M$ we have that $I(m) \in \zones(X)$;
  \item 
    the flow function $F: M {\to} (\Real^{|X|} {\to} \Real^{|X|})$
    is such that for all $m \in M$ we have that $F(m)$ characterizes:
    \[
    \wedge_{x \in X} (\dot{x} = 1), \text{ and }
    \]
  \item $V_0 \in \pred(X)$ is the set of initial valuations is such that
    $V_0 = \wedge_{x \in X} (x=0)$.
  \end{itemize}
\end{definition}

The semantics of timed automata and the concept of timed Kripke structures is
defined is a similar way as for hybrid automata. 

\begin{example}
  The hybrid automaton corresponding to the job-shop scheduling problem, shown
  in Figure~\ref{job2}, can also be modeled as a timed automaton by requiring
  that the rates of variables $x_1$ and $x_2$ is $1$ in all the modes (unlike
  the current example where these clocks are paused in certain modes).
\end{example}

\begin{example}
  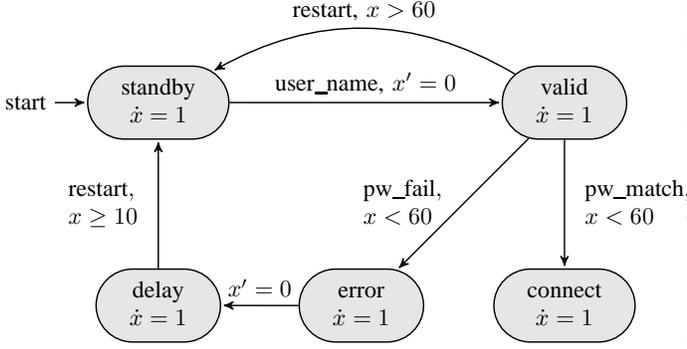
\begin{figure}[t]
\begin{center}
  \scalebox{0.9}{
\begin{tikzpicture}[->,>=stealth',shorten >=1pt,auto,node distance=2.8cm,
  semithick]
  \tikzstyle{every state}=[fill=gray!20!white,minimum size=3em,rounded rectangle]
  
  \node[initial,state] (A) {$\begin{array}{c}\text{standby}\\ \dot{x} = 1 \end{array}$} ;
  
  \node[state] at (6,0) (B) {$\begin{array}{c}\text{valid}\\ \dot{x} = 1 \end{array}$} ;
  \node[state] at (0,-3) (C) {$\begin{array}{c}\text{delay}\\ \dot{x} = 1 \end{array}$} ;
  \node[state] at (3,-3) (D) {$\begin{array}{c}\text{error}\\ \dot{x} = 1 \end{array}$} ;
  \node[state] at (6,-3) (E) {$\begin{array}{c}\text{connect}\\ \dot{x} = 1 \end{array}$} ;
  
  \path (A) edge node [above]{user\_name, $x'=0$} (B);
  \path (B) edge [bend right] node [above]{restart, $x > 60$} (A);

  \path (C) edge node [left]{$\begin{array}{ll} \text{restart}, \\ x\geq10 \end{array}$}(A);
  \path (D) edge node [above]{$x'=0$} (C);
  \path (B) edge node [left]{$\begin{array}{ll} \text{pw\_fail},\\ x<60\end{array}$} (D);
  \path (B) edge node [right]{$\begin{array}{ll} \text{pw\_match},\\ x<60\end{array}$} (E);

\end{tikzpicture}
}
\caption{A time-sensitive login protocol implemented as a timed automaton}
\label{ta}
\end{center}
\end{figure}
  As an example of a timed automaton consider Figure~\ref{ta} that models a
  login protocol using a timed automaton.  
  The system starts in the ``standby'' mode. 
  If the user gives a correct password within $60$ time-units after giving 
  the user name, a connection will be established; if, however, the password
  given is wrong, the system restarts after a delay of at least $10$ time units.
  Moreover, if no password is given within $60$ time units after supplying user
  name, then the system restarts in the standby mode.
  This system is modeled using a timed automaton with five modes and one clock
  in Figure~\ref{ta}. 
\end{example}

Alur and Dill~\cite{AD94} proposed the notion of region equivalence to define
a bisimulation relation over the timed Kripke structures $(\sem{\Tt}, P, L)$.  
We say that two clock valuations $\nu$ and $\nu'$ are \emph{region
  equivalent}, and we write $\nu \sim_R \nu'$, if and only if all
clocks have the same integer parts in $\nu$ and~$\nu'$, and if the
partial orders of the clocks, determined by their fractional parts in 
$\nu$ and $\nu'$, are the same.  

\begin{definition}[Region Equivalence]
  Let $\Tt$ be a timed automaton and let $K$ be the
  maximum constant used in the guards of $\Tt$.
  We say that two clock valuations $\nu$ and $\nu'$ are \emph{region
    equivalent}, and we write $\nu \sim_R \nu'$ if and only if:
  \begin{itemize}
  \item 
    either for $x \in X$ we have $\nu(x) {>} K$ and $\nu'(x)  {>} K$, or 
  \item 
    for any $x, y \in X$ with $\nu(x), \nu'(x) \leq K$ and $\nu(y), \nu'(y) \leq
    K$ the following conditions hold:
    \begin{itemize}
    \item 
     $\floor{\nu(x)} = \floor{\nu'(x)}$, and $\fract{\nu(x)} = 0$ iff
     $\fract{\nu'(x)} = 0$,
    \item 
     $\floor{\nu(y)} = \floor{\nu'(y)}$, and $\fract{\nu(y)} = 0$ iff
     $\fract{\nu'(y)} = 0$,
   \item
     $\fract{\nu(x)}  \leq \fract{\nu(y)}$ if and only if
     $\fract{\nu'(x)}  \leq \fract{\nu'(y)}$,
   \end{itemize}
   where  $\fract{c} \rmdef (c - \floor{c})$ represents the fractional part of
   $c \in \Rplus$. 
 \end{itemize}
\end{definition}

It is easy to see that $\sim_R$ is an equivalence relation. 
For a clock valuation $\nu$ we write $[\nu]$ for the region equivalence class
of $\nu$.
Region equivalence relation can be extended from valuations to configurations
of a timed automaton $\Tt$ in a straightforward manner:
we say that two configurations $(m, \nu)$ and $(m',\nu')$ are region equivalent,
and we write $[(m, \nu)] = [(m', \nu')]$, if and only if $m = m'$ and 
 $[\nu] = [\nu']$. 

Alur and Dill~\cite{AD94} showed that region equivalence relations characterize
finite bisimulation quotients for timed Kripke structures by showing that the
number of equivalence classes for a timed automaton $(M, M_0, \Sigma, X, \Delta,
I, F, V_0)$ are bounded from above by $|M|\cdot|X|!\cdot 2^{|X|}\cdot
\prod_{i=1}^{|X|} \cdot (2.K+2)$.
\begin{theorem}[\cite{AD94}]
  Region equivalence relation characterizes a finite bisimulation quotient for
  timed Kripke structures.  
\end{theorem}

This theorem combined with Theorem~\ref{thm:finiteKSDec} proves the decidability
of LTL model checking for timed Kripke structures.
The complexity of LTL model checking was considered by Courcoubetis and
Yannakakis~\cite{CY92} who showed that simple reachability problem for timed
Kripke structures with three or more clocks is PSPACE-complete.
Despite the high computational complexity of verification, algorithms based on
region equivalence relation coupled with clever data-structures~\cite{UPPAAL} to
symbolically represent sets of regions have been shown to perform well in
practice on medium-sized applications~\cite{UP01,CJLRR09}.   
UPPAAL~\cite{TUppaal}, KRONOS~\cite{TKronos}, and RED~\cite{TRED} are some of
the leading tools that can perform timed automata based verification.
The theory of timed automata has also been extended in several directions
to allow them to model more realistic real-time systems, e.g. real-time systems
with cost and rewards~\cite{LBBFHPR01,BFHLPRV01,BBL08,RLS06,JT08},
uncontrollable nondeterminism~\cite{AMP95,AMPS98,AM99,ABM04,Bou06,BHPR07},
stochastic behavior~\cite{ACD91,KNSS00,Bea03,KNPS06,KNSW07,JLS08,MLK10,BF09},
and recursion~\cite{TW10a,AAS12}.   
For a detailed overview of these extensions we refer to~\cite{Waez13}. 

\subsection{Multi-Rate and Rectangular Hybrid Automata}
\label{sec:rectangular}
Multi-rate hybrid automata, introduced by Henzinger and
Kopke~\cite{HK99,PV94,HKPV98}, are a subclass of hybrid automata where the
dynamics of variables is restricted to constant rates. 
However,  unlike timed
automata, different variables can have different rates, and it can vary among
different modes. 
Moreover, during discrete transitions these variables can be reseted to
real numbers.
Also in a multi-rate hybrid automaton the set of predicates permitted to appear
as guard on transitions is restricted to the following kind of rectangular
predicates:    
\begin{equation}
  g:= c' \bowtie x \bowtie c,\label{guards2}
\end{equation}
where $x$ is a variable, $\bowtie \in \{<, \leq, =, >, \geq\}$ and  
$c, c' \in \mathbb{N}$. 
We write $\rect(X)$ for this class of rectangular predicates over the set $X$. 
Formally, we define a multi-rate hybrid automata as a restriction of hybrid
automata in the following manner.
\begin{definition}[Multi-rate Hybrid Automata: Syntax]
  A multi-rate hybrid automaton is a hybrid automaton 
  $\Hh= (M, M_0, \Sigma, X, \Delta, I, F, V_0)$ with the following restrictions:
  \begin{itemize}
  \item the transition relation 
    $\Delta \subseteq M \times \pred(X) \times \Sigma \times \pred(X\cup X')
    \times M$ is such that if $(m, g, a, j, m') \in \Delta$ then 
    \begin{itemize}
    \item the guard $g$ is of the form (\ref{guards2}), i.e. $g  \in \rect(X)$
      and  
    \item the jump predicate $j$ only permits variable resets to real numbers,
      i.e.  $j$ is of the form 
      \[ 
      \wedge_{x \in Y} (x'{=} c_x)
      \]
      where $Y \subseteq X$ and $c_x \in \Int$ for each $x \in Y$.
      We denote such set $Y$ as $\rst(j)$.
    \end{itemize}
  \item 
    the mode-invariant function $I: M \to \pred(X)$ is such that for all
    $m \in M$ we have that $I(m) \in \rect(X)$;
  \item 
    the flow function $F: M {\to} (\Real^{|X|} {\to} \Real^{|X|})$ 
    is such that for all $m \in M$ we have that $F(m)$ characterize:
    \[ 
    \wedge_{x \in X} (\dot{x} = c_{x,m}),
    \]
    where $c_{x,m} \in \Int$ for each $x \in X$;  and 
  \item $V_0 \in \pred(X)$ is the set of initial valuations is such that
    $V_0 = \wedge_{x \in X} x=0$.
  \end{itemize}
\end{definition}

The semantics of multi-rate automata and the concept of multi-rate Kripke
structures is defined is a similar way as for hybrid automata. 
Rectangular hybrid automata~\cite{HK99,HKPV98} are a generalization of
multi-rate hybrid automata where within each mode the rate of a variable can
change non-deterministically within a given mode-dependent interval.

Using a reduction from two counter Minsky machine, one can easily show that the
LTL model checking problem for multi-rate hybrid automata is undecidable. 
\begin{theorem}[\cite{HKPV98}]
  LTL model-checking problem for multi-rate hybrid automata is undecidable.
\end{theorem}
%\begin{proof}
%  This can be easily seen as follows: 
%  Consider a uninitialized multirate automaton $\Aa$ with variables $x_1, x_2,
%  x_3$. 
%  As done in Theorem \ref{undec}, we can simulate a 2 counter machine $\Mm$ and
%  show that $\Aa$ is empty iff $\Mm$ halts.  
%  Let $x_1, x_2$ encode the counter values $c,d$. 
  
%  Lets discuss how to simulate a decrement instruction $l_i$: decrement $c$,
%  goto $l_j$. 
%  We start in a mode labeled $l_i$, with $x_3=0$, and $x_1, x_2$ holding the
%  current counter values of counters $c,d$. At mode $l_i$, $F(l_i,x_1)=-1,
%  F(l_i,x_2)=0$ and $F(l_i,x_3)=1$.  
%  From $l_i$, we goto mode $l_j$ with constraints $x_1 >0, x_3=1$, and reset
%  $x_3$. It is easy to see that at $l_j$,  the value of $x_2$ is unchanged,
%  $x_3=0$ and $x_1=c-1$. Figure \ref{dec} gives the details. 
%  \input{dec-2}
%\end{proof}

We say that a multi-rate (or rectangular) hybrid automaton is {\it initialized}
if it satisfies the property that every transition between two modes with
different rates (rate intervals, resp.) for a variable, resets that variable,
i.e. for every transition $(m, g, a, j, m') \in \Delta$ with $F(m)(x) \neq
F(m')(x)$ we have $x \in \rst(j)$. 
Figure \ref{rect} shows an initialized rectangular  automaton.
\begin{figure}[t]
\begin{center}
\begin{tikzpicture}[->,>=stealth',shorten >=1pt,auto,node distance=2.8cm,
  semithick]
  \tikzstyle{every state}=[fill=gray!20!white,minimum size=3em,rounded rectangle]
    \node[initial, 
    initial text={},state,fill=gray!30] at (-4,0) (A) {$\dot{x} \in [1,2]$} ;
   \node[state,fill=gray!20] at (-1,0) (B) {$\dot{x} \in [1,2]$} ;
    \node[state,fill=gray!20] at (2,0) (C) {$\dot{x} \in [7,8]$} ;
\path(A) edge [bend left=10] node [above] {$x \leq 10?$} (B);
 \path(B) edge [bend left=10] node [above] {$x'=3$} (C);
\path(C) edge [bend left=10] node [above] {$x'=0$} (B);
 \end{tikzpicture}
\caption{An initialized rectangular automaton}
\label{rect}
\end{center}
\end{figure}
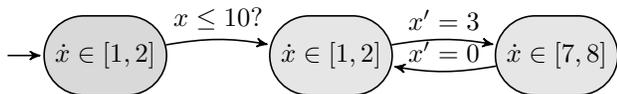

Henzinger et al.~\cite{HKPV98} showed the decidability of initialized
rectangular and multi-rate hybrid automata. 

\begin{theorem}
 \label{bisim:ra}
 The LTL model-checking problem for initialized rectangular (multi-rate) hybrid
 automata is decidable. 
\end{theorem}
\begin{proof}
  The decidability of LTL model-checking problem for initialized multi-rate
  automata by reducing the problem to similar problem for timed automata by
  rescaling the rate of all variables to one via appropriate adjustment of the
  constraints on the mode invariants and guards in all the transitions.  
  
  To prove the decidability for an initialized rectangular automaton $\Hh_r$, we
  reduce the problem to corresponding problem for an initialized multi-rate
  automaton $\Hh_m$. 
  Each variable $x$ of $\Hh_r$ with rate in the rectangle  $a \leq \dot{x} \leq
  b$ is simulated using two variables $x_l, x_u$ such that $\dot{x_l}=a$ and
  $\dot{x_u}=b$.  
  The variables $x_l, x_u$ keep track of the lower and upper bounds of $x$
  respectively.  
  With this replacement, the invariant conditions of modes, as well as guards
  and resets on transitions have to be adjusted appropriately.  
  For example, if we had a transition with guard $x \leq 10$, then it is
  replaced with  (i) $x_l \leq 10$ and (ii)$x_u > 10, x'_u=10$. 
  This conversion from initialized rectangular to initialized multirate automata
  is language preserving. 
  Hence,  from the decidability of LTL model checking problem for initialized
  multi-rate hybrid automata, the decidability for initialized rectangular
  hybrid follows.
\end{proof}

\subsection{Piecewise-Constant Derivative Systems and Their Variants}
\label{sec:MMS}
Asarin, Maler, and Pnueli~\cite{AMP95b} initiated the study of hybrid
dynamical systems with piecewise-constant derivatives (PCD) defined as a partition of
the Euclidean space into a finite set of regions (polyhedral predicates), where
the dynamics in a region is defined by a constant rate vector. 
They defined PCD systems as completely deterministic systems where a discrete
transition occurs at region boundaries, where runs change their directions
according to the rate vector available in the new region.
Given the simplicity of such systems, it is perhaps surprising that the
reachability problem for PCD systems with three or more variables is
undecidable~\cite{AMP95b}.
In fact, Asarin and Maler~\cite{AM98} observed that, due to the capability of
such systems to perform Zeno runs, every set of arithmetical hierarchy (a
hierarchy of undecidable problems) can be recognized by a PCD system of some
finite dimension.  
On the positive side, Asarin, Maler, and Pnueli~\cite{AMP95b} gave an algorithm
to solve the reachability problem for two-dimensional PCD systems. 
Cerans and Viksna~\cite{CV95} later generalized this decidability result to more
general piecewise-Hamiltonian systems.
We also mention the work of Asarin, Schneider, Yovine~\cite{ASY01} who extended
the decidability result for two-dimensional PCD systems to a non-deterministic
setting of simple planar differential inclusion systems (SPDIs) where a number
of rate vectors are available in each region.  

Kesten, Pnueli, Sifakis, and Yovine~\cite{KPSY92} also studied another variant 
of constant-rate hybrid systems, called \emph{integration graphs}, that can be
considered as a subset of multi-rate automaton where no test of  non-clock
(integrator) variables is allowed to appear on a loop.
Kesten et al.~\cite{KPSY92} showed the decidability for the two subclasses of
integration graphs: the class with a single clock variable, and the class where
integrators are tested only once.

Recently, Bouyer et al.~\cite{BFLMS08} introduced timed automata with energy
constraints, that can be considered as multi-rate automata with a single
non-clock variable (energy variable) that does not appear on guards, and showed
decidability of schedulability problem where the energy variable is required to be
greater than a given lower-bound. 
Bouyer, Fahrenberg, Larsen, and Markey~\cite{BFLM10} later generalized this
result to give an EXPTIME algorithm for a subclass where energy variables can
grow exponentially. 

Alur, Trivedi, and Wojtczak recently studied constant-rate multi-mode
systems~\cite{ATW12}, that can be considered as multi-rate automata with the
exception that there is no structure in the automata, i.e. any mode can be used
after any other mode, and there is only a global invariant over variables. 
They showed that reachability and schedulability problems for these systems can
be solved in polynomial time for starting states strictly inside the global
invariant space.
Alur, Trivedi, and Wojtczak also showed that introducing either local invariants
or guards make the reachability problem undecidable. 
Alur et al.~\cite{AFMT13} later studied this problem on a
generalization of constant-rate multi-mode systems to bounded-rate multi-mode
system and showed the decidability of the schedulability problem. 

%\subsection{Other Decidable Variants of Hybrid Automata}
%\label{sec:misc}

%o-minimal
%STORMED
%lazy linear, lazy rectangular
%write a line about priced, probabilistic, etc.

\section{Summary}
\label{sec:conclusion}
In this article we presented hybrid automata for modeling and formal 
verification of cyber-physical systems. 
We begin by showing how hybrid automata naturally combine features from
continuous dynamical systems and discrete finite state machines, and provide an
elegant and expressive model.
This expressiveness, however, comes with a price---the simple reachability problem
for simple subclasses of hybrid automata, like piecewise-constant derivative
systems, turned out to be highly undecidable. 
We discussed a general approach of finding finite bisimulation quotient to show
decidability  of subclasses of hybrid automata, and sketched the proof for the 
decidability for two key subclasses: timed automata and initialized rectangular
hybrid automata.  
Hybrid automata provide an intuitive and semantically unambiguous way to model
cyber-physical systems.
These formalisms provide a rich theory and a mature set of tools,
UPPAAL~\cite{TUppaal}, Kronos~\cite{TKronos}, RED~\cite{TRED},
HyTECH~\cite{THyTech}, and PHAVer~\cite{Tphaver}, able to 
perform automatic verification of systems modeled using them. 
A growing number of case-studies using these tools have shown depromise in
extending the state-of-the-art to industrial-sized examples.

%We have shown some examples of dynamics of cyber-physical systems that allow
%decidable verification problems, and also we showed the boundary between
%decidable and undecidable problems.
%We hope that this survey may shed light on the power and restrictions of hybrid
%automata, and will help the system  designers to model the features accordingly
%using hybrid automata to verify the system.

%What we didn't touch upon: probabilistic extensions, recursion, uncontrollable
%non-determinism, partial-information observability,  etc.

%Why Hybrid verification did not catch up with industry like hardware and
%software verification ..

%What can we do to overcome this problem....

\bibliographystyle{plain}
\bibliography{papers}

\begin{thebibliography}{10}

\bibitem{TM0}
{A}{C}{M} {T}uring award citation for {A.} {P}nueli.
\newblock
  {\footnotesize{\url{http://awards.acm.org/citation.cfm?id=4725172&srt=alpha&alpha=P&aw=140&ao=AMTURING&yr=1996}}},
  1996.
\newblock For seminal work introducing temporal logic into computing science
  and for outstanding contributions to program and system verification.

\bibitem{TM1}
{A}{C}{M} {T}uring award citation for {E}. {M}. {C}larke, {E}. {A}. {E}merson,
  and {J}. {S}ifakis.
\newblock
  {\footnotesize{\url{http://awards.acm.org/citation.cfm?id=1167964&srt=alpha&alpha=C&aw=140&ao=AMTURING&yr=2007}}},
  2007.
\newblock For their role in developing Model-Checking into a highly effective
  verification technology, widely adopted in the hardware and software
  industries.

\bibitem{PK0}
{A}{C}{M} {K}anellakis theory and practice award citation for {G}. {J}.
  {H}olzmann, {R}. {P}. {K}urshan, {M}. {Y}. {V}ardi, and {P}. {W}olpe.
\newblock
  {\footnotesize{\url{http://awards.acm.org/citation.cfm?id=1625680&srt=all&aw=147&ao=KANELLAK&yr=2005}}},
  2005.
\newblock For the development of automata-theoretic techniques for
  reactive-systems verification, and the practical realization of powerful
  formal-verification tools based on these techniques.

\bibitem{PK1}
{A}{C}{M} {K}anellakis theory and practice award citation for {R}. {E}.
  {B}ryant, {E}. {M}. {C}larke, {E}. {A}. {E}merson, {K}. {L}. {M}cmillan.
\newblock
  {\footnotesize{\url{http://awards.acm.org/citation.cfm?id=1167964&srt=all&aw=147&ao=KANELLAK&yr=1998}}},
  1998.
\newblock For their invention of "symbolic model checking".

\bibitem{AAS12}
P.~A. Abdulla, M.~F. Atig, and J.~Stenman.
\newblock Dense-timed pushdown automata.
\newblock In {\em Proceedings of the 2012 27th Annual IEEE/ACM Symposium on
  Logic in Computer Science}, LICS '12, pages 35--44, Washington, DC, USA,
  2012. IEEE Computer Society.

\bibitem{AB06}
R.~Alur and M.~Bernadsky.
\newblock Bounded model checking for gsmp models of stochastic real-time
  systems.
\newblock In {\em Proceedings of HSCC}, volume 3927, pages 19--33. LNCS, 2006.

\bibitem{ABM04}
R.~Alur, M.~Bernadsky, and P.~Madhusudan.
\newblock Optimal reachability for weighted timed games.
\newblock In J.~D\'iaz, J.~Karhum\"aki, A.~Lepist\"o, and D~Sannella, editors,
  {\em Proc. ICALP'04}, volume 3142 of {\em LNCS}, pages 122--133. Springer,
  2004.

\bibitem{ACD91}
R.~Alur, C.~Courcoubetis, and D.~L. Dill.
\newblock Model-checking for probabilistic real-time systems.
\newblock In {\em In Automata, Languages and Programming: Proceedings of the
  18th ICALP, Lecture Notes in Computer Science 510}, pages 115--126. Springer,
  1991.

\bibitem{ACHH93}
R.~Alur, C.~Courcoubetis, T.~A. Henzinger, and P.-H. Ho.
\newblock Hybrid automata: An algorithmic approach to the specification and
  verification of hybrid systems.
\newblock In {\em Hybrid Systems I}, volume 736 of {\em Lecture Notes in
  Computer Science}, pages 209--229. Springer-Verlag, 1993.

\bibitem{AD90}
R.~Alur and D.~Dill.
\newblock Automata for modeling real-time systems.
\newblock In {\em International Colloquium on Automata, Languages and
  Programming (ICALP)}, volume 443 of {\em LNCS}, pages 322--335. Springer,
  1990.

\bibitem{AD94}
R.~Alur and D.~Dill.
\newblock A theory of timed automata.
\newblock {\em Theoretical Computer Science}, 126:183--235, 1994.

\bibitem{AFMT13}
R.~Alur, V.~Forejt, S.~Moarref, and A.~Trivedi.
\newblock Safe schedulability of bounded-rate multi-mode systems.
\newblock In {\em Hybrid Systems: Computation and Control ({HSCC})}, pages
  243--252, 2013.

\bibitem{ATW12}
R.~Alur, A.~Trivedi, and D.~Wojtczak.
\newblock Optimal scheduling for constant-rate multi-mode systems.
\newblock In {\em Hybrid Systems: Computation and Control ({HSCC})}, 2012.
\newblock To appear.

\bibitem{AH92}
Rajeev Alur and Thomas~A Henzinger.
\newblock Logics and models of real time: A survey.
\newblock In {\em Real-Time: Theory in Practice}, pages 74--106. Springer,
  1992.

\bibitem{ADMP01}
P.~Arg\'{o}n, G.~Delzanno, S.~Mukhopadhyay, and A.~Podelski.
\newblock Model checking communication protocols.
\newblock In {\em Proceedings of the 28th Conference on Current Trends in
  Theory and Practice of Informatics Piestany: Theory and Practice of
  Informatics}, SOFSEM '01, pages 160--170, London, UK, 2001. Springer-Verlag.

\bibitem{AM98}
E.~Asarin and O.~Maler.
\newblock Achilles and the tortoise climbing up the arithmetical hierarchy.
\newblock {\em Journal of Computer and System Sciences}, 57(3):389--398, 1998.

\bibitem{AM99}
E.~Asarin and O.~Maler.
\newblock As soon as possible: Time optimal control for timed automata.
\newblock In F.~W. Vaandrager and J.~H. van Schuppen, editors, {\em Proc.
  HSCC'99}, volume 1569 of {\em LNCS}, pages 19--30. Springer, 1999.

\bibitem{AMP95b}
E.~Asarin, O.~Maler, and A.~Pnueli.
\newblock Reachability analysis of dynamical systems having piecewise-constant
  derivatives.
\newblock {\em Theoretical Computer Science}, 138(1):35 -- 65, 1995.

\bibitem{AMP95}
E.~Asarin, O.~Maler, and A.~Pnueli.
\newblock Symbolic controller synthesis for discrete and timed systems.
\newblock In P.~Antsaklis, W.~Kohn, A.~Nerode, and S.~Sastry, editors, {\em
  Hybrid Systems II}, volume 999 of {\em LNCS}, pages 1--20. Springer, 1995.

\bibitem{AMPS98}
E.~Asarin, O.~Maler, A.~Pnueli, and J.~Sifakis.
\newblock Controller synthesis for timed automata.
\newblock In P.~Antsaklis, W.~Kohn, A.~Nerode, and Sastry S., editors, {\em
  Proceedings of IFAC Symposium on System Structure and Control}, pages
  469--474. Elsevier, 1998.

\bibitem{ASY01}
E.~Asarin, G.~Schneider, and S.~Yovine.
\newblock On the decidability of the reachability problem for planar
  differential inclusions.
\newblock In {\em Hybrid Systems: Computation and Control}, pages 89--104.
  Springer, 2001.

\bibitem{BK08}
C.~Baier and J.-P. Katoen.
\newblock {\em Principles of Model Checking (Representation and Mind Series)}.
\newblock The MIT Press, 2008.

\bibitem{BLR11}
T.~Ball, V.~Levin, and S.~K. Rajamani.
\newblock A decade of software model checking with slam.
\newblock {\em Communications of the {A}{C}{M}}, 54(7):68--76, July 2011.

\bibitem{BCM11}
D.~Basin, C.~Cremers, and C.~Meadows.
\newblock Model checking security protocols.
\newblock {\em Handbook of Model Checking}, 2011.
\newblock \url{http://people. inf. ethz. ch/cremersc/publications/index. html}.

\bibitem{Bea03}
D.~Beauquier.
\newblock Probabilistic timed automata.
\newblock {\em Theoretical Computer Science}, 292(1):65--84, 2003.

\bibitem{BFHLPRV01}
G.~Behrmann, A.~Fehnker, T.~Hune, K.~G. Larsen, P.~Pettersson, J.~Romijn, and
  F.~W. Vaandrager.
\newblock Minimum-cost reachability for priced timed automata.
\newblock In M.~D. Di~Benedetto and A.~L. Sangiovanni-Vincentelli, editors,
  {\em Proc. HSCC'01}, volume 2034 of {\em LNCS}, pages 147--161, Heidelberg,
  2001. Springer.

\bibitem{UPPAAL}
J.~Bengtsson and W.~Yi.
\newblock Timed automata: Semantics, algorithms and tools.
\newblock In {\em Lectures on Concurrency and Petri Nets}, pages 87--124, 2003.

\bibitem{JJK06}
J.~Berendsen, D.~Jansen, and J-P. Katoen.
\newblock Probably on time and within budget - on reachability in priced
  probabilistic timed automata.
\newblock In {\em Proc. QEST'06}, pages 311--322, Washington, DC, USA, 2006.
  IEEE.

\bibitem{BCCSZ03}
A.~Biere, A.~Cimatti, E.~M. Clarke, O.~Strichman, and Y.~Zhu.
\newblock Bounded model checking.
\newblock {\em Advances in computers}, 58:117--148, 2003.

\bibitem{BCCSZ09}
A.~Biere, A.~Cimatti, E.~M. Clarke, O.~Strichman, and Y.~Zhu.
\newblock Bounded model checking.
\newblock {\em Handbook of Satisfiability}, 185:457--481, 2009.

\bibitem{Bou06}
P.~Bouyer.
\newblock Weighted timed automata: Model-checking and games.
\newblock {\em Electronic Notes Theoretical Computer Science}, 158:3--17, 2006.

\bibitem{BBJLR08}
P.~Bouyer, T.~Brihaye, M.~Jurdzinski, R.~Lazic, and M.~Rutkowski.
\newblock Average-price and reachability-price games on hybrid automata with
  strong resets.
\newblock In {\em Formal Modeling and Analysis of Timed Systems (FORMATS)},
  volume 5215 of {\em LNCS}, pages 63--77. 2008.

\bibitem{BBL08}
P.~Bouyer, E.~Brinksma, and K.~G. Larsen.
\newblock Optimal infinite scheduling for multi-priced timed automata.
\newblock {\em Formal Methods in System Design}, 32(1):3--23, 2008.

\bibitem{BFLM10}
P.~Bouyer, U.~Fahrenberg, K.~G. Larsen, and N.~Markey.
\newblock Timed automata with observers under energy constraints.
\newblock In {\em Hybrid Systems: Computation and Control ({HSCC})}, 2010.

\bibitem{BFLMS08}
P.~Bouyer, U.~Fahrenberg, K.~G. Larsen, N.~Markey, and J.~Srba.
\newblock Infinite runs in weighted timed automata with energy constraints.
\newblock In {\em Formal Modeling and Analysis of Timed Systems}, pages 33--47.
  Springer, 2008.

\bibitem{BF09}
P.~Bouyer and V.~Forejt.
\newblock Reachability in stochastic timed games.
\newblock In {\em International Colloquium on Automata, Languages and
  Programming (ICALP)}, volume 5556 of {\em LNCS}, pages 103--114. Springer,
  2009.

\bibitem{BBR04}
T.~Brihaye, V.~Bruyere, and J.-F. Raskin.
\newblock Model-checking for weighted timed automata.
\newblock In Yassine Lakhnech and Sergio Yovine, editors, {\em Formal
  Techniques, Modelling and Analysis of Timed and Fault-Tolerant Systems},
  volume 3253 of {\em Lecture Notes in Computer Science}, pages 277--292.
  Springer Berlin Heidelberg, 2004.

\bibitem{BHPR07}
T.~Brihaye, T.~A. Henzinger, V.~S. Prabhu, and J.~Raskin.
\newblock Minimum-time reachability in timed games.
\newblock In {\em Proc. ICALP'07}, volume 4596 of {\em LNCS}, pages 825--837.
  Springer, 2007.

\bibitem{BKPS07}
M.~Broy, I.H. Kruger, A.~Pretschner, and C.~Salzmann.
\newblock Engineering automotive software.
\newblock {\em Proceedings of the IEEE}, 95(2):356--373, 2007.

\bibitem{Bu62}
J.~R. B\"uchi.
\newblock On a decision method in restricted second-order arithmetic.
\newblock In {\em Int. Congr. for Logic Methodology and Philosophy of Science},
  pages 1--11. Standford University Press, Stanford, 1962.

\bibitem{CJLRR09}
F.~Cassez, J.~Jessen, K.~Larsen, J.~Raskin, and P.~Reynier.
\newblock Automatic synthesis of robust and optimal controllers: an industrial
  case study.
\newblock In F.~W. Vaandrager and J.~H. van Schuppen, editors, {\em HSCC'09},
  volume 5469 of {\em LNCS}, 2009.

\bibitem{CV95}
K.~Cerans and J.~Viksna.
\newblock Deciding reachability for planar multi-polynomial systems.
\newblock In {\em Hybrid Systems}, pages 389--400, 1995.

\bibitem{GNRR93}
Zhou Chaochen, Anders~P. Ravn, and Michael~R. Hansen.
\newblock An extended duration calculus for hybrid real-time systems.
\newblock In {\em Hybrid Systems}, volume 736 of {\em Lecture Notes in Computer
  Science}, pages 36--59. Springer Berlin Heidelberg, 1993.

\bibitem{C13}
ROBERT~N. CHARETTE.
\newblock This car runs on code.
\newblock
  \url{http://news.discovery.com/autos/toyota-recall-software-code.htm},
  February 2013.

\bibitem{CHP08}
K.~Chatterjee, T.~A. Henzinger, and V.~Prabhu.
\newblock Timed parity games: Complexity and robustness.
\newblock In {\em Proceedings of the Sixth International Conference on Formal
  Modeling and Analysis of Timed Systems (FORMATS'08)}, volume 5215 of {\em
  LNCS}, pages 124--140, 2008.

\bibitem{CAJS98}
E.~M. Clarke, E.~A. Emerson, S.~Jha, and A.~P. Sistla.
\newblock Symmetry reductions in model checking.
\newblock In {\em Computer Aided Verification}, pages 147--158. Springer, 1998.

\bibitem{CES09}
E.~M. Clarke, E.~A. Emerson, and J.~Sifakis.
\newblock Model checking: algorithmic verification and debugging.
\newblock {\em Communications of the {A}{C}{M}}, 52(11):74--84, 2009.

\bibitem{CFHKOST03}
E.~M. Clarke, A.~Fehnker, Z.~Han, J.~Krogh, B. H.and~Ouaknine, O.~Stursberg,
  and M.~Theobald.
\newblock Abstraction and counterexample-guided refinement in model checking of
  hybrid systems.
\newblock {\em International Journal of Foundations of Computer Science},
  14(4):583--604, 2003.

\bibitem{CGHLV00}
E.~M. Clarke, O.~Grumberg, S.~Jha, Y.~Lu, and H.~Veith.
\newblock Counterexample-guided abstraction refinement.
\newblock In {\em CAV}, pages 154--169, 2000.

\bibitem{CGP99}
E.~M. Clarke, O.~Grumberg, and D.~Peled.
\newblock {\em Model Checking}.
\newblock {MIT} Press, 1999.

\bibitem{CY92}
C.~Courcoubetis and M.~Yannakakis.
\newblock Minimum and maximum delay problems in real-time systems.
\newblock In {\em Formal Methods in System Design}, volume~1, pages 385--415,
  Dordrecht, 1992. Kluwer.

\bibitem{dAFHMS03}
L.~de~Alfaro, M.~Faella, T.~A. Henzinger, R.~Majumdar, and M.~Stoelinga.
\newblock The element of surprise in timed games.
\newblock In R.~Amadio and D.~Lugiez, editors, {\em Proc. CONCUR'03}, volume
  2761 of {\em LNCS}, pages 144--158. Springer, 2003.

\bibitem{Eme96}
E.~A. Emerson.
\newblock Model checking and mu-calculus.
\newblock In N.~Immerman and Ph.~G. Kolaitis, editors, {\em Descriptive
  Complexity and Finite Models}, volume~31 of {\em {DIMACS} Series in Discrete
  Mathematics and Theoretical Computer Science}, pages 185--214. American
  Mathematical Society, 1996.

\bibitem{Frehse05}
G.~Frehse.
\newblock Phaver: Algorithmic verification of hybrid systems past hytech.
\newblock pages 258--273. Springer, 2005.

\bibitem{FTY11}
Hong Fu, Guangyu Tian, Quanshi Chen, and Yiding Jin.
\newblock Hybrid automata of an integrated motor-transmission powertrain for
  automatic gear shift.
\newblock In {\em American Control Conference (ACC), 2011}, pages 4604--4609,
  2011.

\bibitem{GTW02}
E.~Gr{\"a}del, W.~Thomas, and T.~Wilke, editors.
\newblock {\em Automata, Logics, and Infinite Games. A Guide to Current
  Research}, volume 2500 of {\em LNCS}.
\newblock Springer, 2002.

\bibitem{HHW97}
T.~A. Henzinger, P.~Ho, and H.~Wong-toi.
\newblock Hytech: A model checker for hybrid systems.
\newblock {\em Software Tools for Technology Transfer}, 1:460--463, 1997.

\bibitem{HK99}
T.~A. Henzinger and P.~W. Kopke.
\newblock Discrete-time control for rectangular hybrid automata.
\newblock {\em Theor. Comput. Sci.}, 221(1-2):369--392, 1999.

\bibitem{HKPV98}
T.~A. Henzinger, P.~W. Kopke, A.~Puri, and P.~Varaiya.
\newblock What's decidable about hybrid automata?
\newblock {\em Journal of Computer and System Sciences}, 57(1):94 -- 124, 1998.

\bibitem{THyTech}
Hytech.
\newblock
  \footnotesize{\url{http://embedded.eecs.berkeley.edu/research/hytech/}}.

\bibitem{JPM12}
Zhihao Jiang, Miroslav Pajic, and Rahul Mangharam.
\newblock Cyber--physical modeling of implantable cardiac medical devices.
\newblock {\em Proceedings of the IEEE}, 100(1):122--137, 2012.

\bibitem{JLS08}
M.~Jurdzi{\'n}ski, J.~Sproston, and F.~Laroussinie.
\newblock Model checking probabilistic timed automata with one or two clocks.
\newblock {\em Logical Methods in Computer Science}, 4(3):12, 2008.

\bibitem{JT08}
M.~Jurdzi{\'n}ski and A.~Trivedi.
\newblock Concavely-priced timed automata.
\newblock In F.~Cassez and C.~Jard, editors, {\em Formal Modeling and Analysis
  of Timed Systems (FORMATS)}, volume 5215 of {\em LNCS}, pages 48--62.
  Springer, 2008.

\bibitem{KPSY92}
Y.~Kesten, A.~Pnueli, J.~Sifakis, and Yovine. S.
\newblock Integration graphs: A class of decidable hybrid systems.
\newblock In R.~L. Grossman, A.~Nerode, A.~P. Ravn, and H.~Rischel, editors,
  {\em Hybrid Systems}, volume 736 of {\em LNCS}, pages 179--208. Springer,
  1992.

\bibitem{koy90}
Ron Koymans.
\newblock Specifying real-time properties with metric temporal logic.
\newblock {\em Real-time systems}, 2(4):255--299, 1990.

\bibitem{TKronos}
Kronos.
\newblock \footnotesize{\url{http://www-verimag.imag.fr/TEMPORISE/kronos/}}.

\bibitem{Kur08}
R.~P. Kurshan.
\newblock Verification technology transfer.
\newblock In Orna Grumberg and Helmut Veith, editors, {\em 25 Years of Model
  Checking}, pages 46--64. Springer-Verlag, Berlin, Heidelberg, 2008.

\bibitem{KNPS06}
M.~Kwiatkowska, G.~Norman, D.~Parker, and J.~Sproston.
\newblock Performance analysis of probabilistic timed automata using digital
  clocks.
\newblock {\em FMSD}, 29:33--78, 2006.

\bibitem{KNSS00}
M.~Kwiatkowska, G.~Norman, R.~Segala, and J.~Sproston.
\newblock Verifying quantitative properties of continuous probabilistic timed
  automata.
\newblock In {\em Proc. of 11th International Conference on Concurrency
  Theorey, (CONCUR'00)}, volume 1877 of {\em LNCS}, pages 123--137, 2000.

\bibitem{KNSW07}
M.~Kwiatkowska, G.~Norman, J.~Sproston, and F.~Wang.
\newblock Symbolic model checking for probabilistic timed automata.
\newblock {\em Information and Computation}, 205(7):1027--1077, 2007.

\bibitem{KNSS99}
M.~Z. Kwiatkowska, G.~Norman, R.~Segala, and J.~Sproston.
\newblock Automatic verification of real-time systems with discrete probability
  distributions.
\newblock In {\em ARTS}, pages 75--95, 1999.

\bibitem{Lam77}
L.~Lamport.
\newblock Proving the correctness of multiprocess programs.
\newblock {\em Software Engineering, IEEE Transactions on}, SE-3(2):125--143,
  1977.

\bibitem{LBBFHPR01}
K.~G. Larsen, G.~Behrmann, E.~Brinksma, A.~Fehnker, T.~Hune, P.~Pettersson, and
  J.~Romijn.
\newblock As cheap as possible: Efficient cost-optimal reachability for priced
  timed automata.
\newblock In G.~Berry, H.~Comon, and A.~Finkel, editors, {\em Proc. CAV'01},
  volume 2102 of {\em LNCS}, pages 493--505, Heidelberg, 2001. Springer.

\bibitem{LTS08}
J.~Lygeros, C.~Tomlin, and S.~Sastry.
\newblock {\em Hybrid Systems: Modeling, Analysis and Control}.
\newblock In preparation., 2008.
\newblock Unpublished anuscript,
  \url{http://www-inst.cs.berkeley.edu/~ee291e/sp09/handouts/book.pdf}.

\bibitem{MLK10}
O.~Maler, K.~G. Larsen, and B.~Krogh.
\newblock On zone-based analysis of duration probabilistic automata.
\newblock In {\em INFINITY}, pages 33--46, 2010.

\bibitem{MP87}
Z.~Manna and A.~Pnueli.
\newblock A hierarchy of temporal properties.
\newblock In {\em Proceedings of the sixth annual ACM Symposium on Principles
  of distributed computing}, PODC '87, pages 205--205, New York, NY, USA, 1987.
  ACM.

\bibitem{MP92}
Z.~Manna and A.~Pnueli.
\newblock {\em The Temporal Logic of Reactive and Concurrent Systems}.
\newblock Springer-Verlag, 1992.

\bibitem{MCJ97}
W.~Marrero, E.~M. Clarke, and S.~Jha.
\newblock Model checking for security protocols.
\newblock Technical report, DTIC Document, 1997.

\bibitem{Matlab}
MATLAB.
\newblock {\em version 7.10.0 (R2010a)}.
\newblock The MathWorks Inc., Natick, Massachusetts, 2010.

\bibitem{McMilan93}
Kenneth~L McMillan.
\newblock {\em Symbolic model checking}.
\newblock Springer, 1993.

\bibitem{Min67}
Marvin~L. Minsky.
\newblock {\em Computation: finite and infinite machines}.
\newblock Prentice-Hall, Inc., 1967.

\bibitem{TNuSMV}
Nusmv.
\newblock \footnotesize{\url{http://nusmv.fbk.eu/}}.

\bibitem{OW08}
J.~Ouaknine and J.~Worrell.
\newblock Some recent results in metric temporal logic.
\newblock In Franck Cassez and Claude Jard, editors, {\em Formal Modeling and
  Analysis of Timed Systems}, volume 5215 of {\em Lecture Notes in Computer
  Science}, pages 1--13. Springer Berlin Heidelberg, 2008.

\bibitem{Pas06}
William Pasillas-L{\'e}pine.
\newblock Hybrid modeling and limit cycle analysis for a class of five-phase
  anti-lock brake algorithms.
\newblock {\em Vehicle System Dynamics}, 44(2):173--188, 2006.

\bibitem{Peled94}
Doron Peled.
\newblock Combining partial order reductions with on-the-fly model-checking.
\newblock In {\em Computer aided verification}, pages 377--390. Springer, 1994.

\bibitem{Tphaver}
Phaver.
\newblock \footnotesize{\url{http://www-verimag.imag.fr/~frehse/phaver_web/}}.

\bibitem{PV94}
A.~Puri and P.~Varaiya.
\newblock Decidability of hybrid systems with rectangular differential
  inclusion.
\newblock In {\em Computer Aided Verification (CAV)}, pages 95--104, 1994.

\bibitem{RLS06}
J.~I. Rasmussen, K.~G. Larsen, and K.~Subramani.
\newblock On using priced timed automata to achieve optimal scheduling.
\newblock {\em Formal Methods in System Design}, 29(1):97--114, 2006.

\bibitem{TRED}
{RED}.
\newblock {\footnotesize{\url{http://cc.ee.ntu.edu.tw/\~farn/red/}}}.

\bibitem{Roy09}
Jaijeet Roychowdhury.
\newblock Numerical simulation and modelling of electronic and biochemical
  systems.
\newblock {\em Foundations and Trends in Electronic Design Automation},
  3(2---3):97--303, February 2009.

\bibitem{SC85}
A.~P. Sistla and E.~M. Clarke.
\newblock The complexity of propositional linear temporal logics.
\newblock {\em Journal of the ACM}, 32(3):733--749, July 1985.

\bibitem{TSPIN}
Spin.
\newblock \footnotesize{\url{http://spinroot.com/}}.

\bibitem{SMF97}
T.~Stauner, O.~Muller, and M.~Fuchs.
\newblock Using hytech to verify an automotive control system.
\newblock In Oded Maler, editor, {\em Hybrid and Real-Time Systems}, volume
  1201 of {\em Lecture Notes in Computer Science}, pages 139--153. Springer
  Berlin Heidelberg, 1997.

\bibitem{TW10a}
A.~Trivedi and D.~Wojtczak.
\newblock Recursive timed automata.
\newblock In {\em Proceedings of the 8th International Symposium on Automated
  Technology for Verification and Analysis}, volume 6252 of {\em LNCS}, pages
  306--324. Springer-Verlag, September 2010.

\bibitem{UP01}
Uppaal case-studies.
\newblock
  {\footnotesize{\url{http://www.it.uu.se/research/group/darts/uppaal/examples.shtml}}}.

\bibitem{TUppaal}
Uppaal.
\newblock \footnotesize{\url{http://www.uppaal.com/}}.

\bibitem{Waez13}
Md~Tawhid~Bin Waez, Juergen Dingel, and Karen Rudie.
\newblock A survey of timed automata for the development of real-time systems.
\newblock {\em Computer Science Review}, 2013.

\bibitem{Mathematica}
Inc. Wolfram~Research.
\newblock {\em Mathematica Edition: Version 8.0}.
\newblock Wolfram Research, Inc., 2010.
\newblock Champaign, Illinois.

\bibitem{WVS83}
P.~Wolper, M.Y. Vardi, and A.~P. Sistla.
\newblock Reasoning about infinite computation paths.
\newblock In {\em Foundations of Computer Science}, pages 185--194, 1983.

\end{thebibliography}

\end{document}